\documentclass[11pt]{article}
\usepackage[margin=1in]{geometry}
\usepackage[colorlinks=true,
            linkcolor=blue,
            citecolor=blue,
            urlcolor=blue]{hyperref}

\usepackage{fixmath}
\usepackage{bm}
\usepackage{amsbsy}
\usepackage{color}
\usepackage{verbatim}
\usepackage{multirow}
\usepackage{amssymb}
\usepackage{amsthm}
\usepackage{array}
\usepackage{mathtools}
\usepackage{amsmath}
\usepackage{graphicx}
\usepackage{tikz}
\usetikzlibrary{positioning}
\usepackage{xcolor}

\usepackage{booktabs}
\usepackage{subcaption}
\usepackage{thm-restate}

\newtheorem{theorem}{Theorem}
\newtheorem{lemma}{Lemma}

\newtheorem{definition}{Definition}

\newtheorem{property}{Property}

\newtheorem{claim}{Claim}

\usepackage{algorithm}
\usepackage{algorithmicx}
\usepackage{algpseudocode}

\usepackage{authblk}

\usepackage{acro}

\newcommand{\size}[1]{\ensuremath{|#1|}}

\def\el{\emph{et al.}}

\def\A{\mathcal{A}}
\def\B{\mathcal{B}}

\def\F{\mathcal{F}}

\def\I{\mathcal{I}}

\def\M{\mathcal{M}}

\def\O{\mathcal{O}}
\def\P{\mathcal{P}}
\def\R{\mathcal{R}}
\def\T{\mathcal{T}}

\def\V{\mathcal{V}}

\def\OPT{\mbox{OPT}}

\title{Parameterized Approximation Algorithms for TSP on Non-Metric Graphs}

\author[1,2]{Jingyang Zhao\thanks{Email: \texttt{jingyangzhao1020@gmail.com}}}
\author[3]{Zimo Sheng\thanks{Email: \texttt{shengzimo2016@gmail.com}}}
\author[1]{Mingyu Xiao\thanks{Email: \texttt{myxiao@uestc.edu.cn}}}

\affil[1]{University of Electronic Science and Technology of China}
\affil[2]{Kyung Hee University, Yongin-si, South Korea}
\affil[3]{Anhui University}

\date{}

\begin{document}

\maketitle

\begin{abstract}
The Traveling Salesman Problem (TSP) is a classic and extensively studied problem with numerous real-world applications in artificial intelligence and operations research. It is well-known that TSP admits a constant approximation ratio on metric graphs but becomes NP-hard to approximate within any computable function $f(n)$ on general graphs. This disparity highlights a significant gap between the results on metric graphs and general graphs. Recent research has introduced some parameters to measure the ``distance'' of general graphs from being metric and explored Fixed-Parameter Tractable (FPT) approximation algorithms parameterized by these parameters. Two commonly studied parameters are $p$, the number of vertices in triangles violating the triangle inequality, and $q$, the minimum number of vertices whose removal results in a metric graph. In this paper, we present improved FPT approximation algorithms with respect to these two parameters. For $p$, we propose an FPT algorithm with a 1.5-approximation ratio, improving upon the previous ratio of 2.5. For $q$, we significantly enhance the approximation ratio from 11 to 3, advancing the state of the art in both cases. In addition, when $p$ (or $q$) is a constant, we obtain a better approximation ratio.
\end{abstract}

\maketitle

\section{Introduction}
The Traveling Salesman Problem (TSP) is a well-known optimization problem with numerous real-world applications~\cite{applegate2006traveling}, including logistics, genetics, manufacturing, and telecommunications. Given an edge-weighted complete graph $G=(V,E,w)$ with $n$ vertices and a non-negative edge weight function $w:E\to\mathbb{R}_{\geq 0}$, the objective is to find a minimum-weight TSP tour (i.e., Hamiltonian cycle) that visits each vertex exactly once. Due to its NP-hardness~\cite{garey1979computers}, TSP and its variants have been extensively studied in the context of approximation algorithms~\cite{DBLP:journals/corr/abs-2311-00604,approxtsp}.

While TSP on general graphs is NP-hard to approximate within any computable function $f(n)$~\cite{sahni1976p}, its widely-used variant on metric graphs (metric TSP), where the weight function satisfies the triangle inequality, admits a constant approximation ratio. Notably, Christofides~\cite{christofides1976worst} and Serdyukov~\cite{serdyukov1978some} independently proposed the famous $1.5$-approximation algorithm for metric TSP. Later, Karlin \el~\cite{DBLP:conf/stoc/KarlinKG21,DBLP:conf/ipco/KarlinKG23} improved the ratio to $1.5-10^{-36}$.

Given the significant gap between the results on general graphs and metric graphs, many approximation algorithms have been developed for graphs that are ``nearly'' metric, aiming to measure the ``distance'' from approximability. In the literature, two main research directions are explored. The first direction relaxes the condition for the triangle inequality to hold, for instance, by satisfying the \emph{$\tau$-triangle inequality}: $\tau\cdot (w(a,b)+w(b,c))\geq w(a,c)$ for all $a,b,c\in V$. The second direction relaxes the scope of the triangle inequality, where it is only required to hold for a subset of the vertices.

For the first direction, a comprehensive research framework has been established (see the survey~\cite{klasing2018modern}). Andreae and Bandelt~\cite{andreae1995performance} gave a $(3\tau^2+\tau)/2$-approximation algorithm for TSP. Bender and Chekuri~\cite{bender2000performance} improved this ratio to $4\tau$ when $\tau>7/3$, and M{\"o}mke~\cite{momke2015improved} further improved the ratio to $3(\tau^2+\tau)/4$ for small $\tau$. More improvements can be found in~\cite{andreae2001traveling,bockenhauer2002towards}.

For the second direction, two lines of research have made explorations. First, for a graph where the vertex set can be partitioned into two subsets $V_1$ and $V_2$ such that both $G[V_1]$ and $G[V_2]$ are metric, Mohan~\cite{mohan2017constant} proposed a $3.5$-approximation algorithm. Further, when the vertex set is partitioned into $k$ groups and the tour must visit vertices within each group consecutively, the problem is known as \emph{Subgroup Path Planning}. In this problem, the triangle inequality $w(a,b)+w(b,c)\geq w(a,c)$ may be violated only when vertices $a$ and $c$ belong to the same group while $b$ belongs to a different one. This problem has applications in AI, such as for polishing robots~\cite{DBLP:conf/ijcai/SafilianHES16}. Sumita~\el~\cite{DBLP:conf/ijcai/SumitaYKK17} first proposed a $3$-approximation algorithm, which was recently improved to $2.167$~\cite{DBLP:conf/aaai/000100X25}. Moreover, when each group has a size of at most 2, the problem reduces to \emph{Subpath Planning}~\cite{DBLP:conf/ijcai/SafilianHES16}, which has applications in electronic printing and admits a $1.5$-approximation ratio~\cite{DBLP:conf/ijcai/SumitaYKK17}.

In fact, a natural direction for relaxing the scope of the triangle inequality is to consider graphs where only a small number of vertices are involved in triangles that violate the inequality. A practical example is urban tour bus routing. Operators design TSP tours to visit landmarks, but exclusive express routes between certain high-profile landmarks create shortcuts that violate the triangle inequality. These violations are typically caused by a limited number of landmarks, meaning the metric property could be restored by removing just a few critical nodes.

Recently, this direction has been investigated using techniques from parameterized complexity. For a minimization problem, an algorithm is called an Fixed-Parameter Tractable (FPT) $\rho$-approximation algorithm if, for any instance $\I$ with a parameter $k$, it runs in $\O(f(k))\cdot poly(\size{\I})$ time and returns a $\rho$-approximate solution, where $f(k)$ is a computable function of $k$. Such algorithms are particularly useful when $k$ is small~\cite{cygan2015parameterized}. Following this line, Zhou~\el~\cite{ZhouLG22,ZhouLG22+} introduced two natural parameters: the number of triangles $p'$ that violate the triangle inequality, and the minimum number of vertices $q$ whose removal results in a metric graph. 
For TSP parameterized by $p'$, they presented an FPT $3$-approximation algorithm, which was recently improved to $2.5$ by Bampis \el~\cite{arxiv24}. 
For TSP parameterized by $q$, they gave an FPT $11$-approximation algorithm.

Our research follows this line to study FPT approximation algorithms for TSP with respect to the parameters $p$ and $q$, where $p$ denotes the number of vertices involved in triangles that violate the triangle inequality.
Although $p' = \Omega(p)$, existing FPT approximation algorithms for TSP with respect to $p'$ also extend directly to $p$, as noted by Bampis \el~\cite{arxiv24}.
Moreover, FPT approximation algorithms with respect to $p$ and $q$ have also been studied for several location problems~\cite{DBLP:conf/cocoon/ZhouZWG25}.

\subsection{Our Results}
In this paper, we present improved FPT approximation algorithms for TSP with respect to the parameters $p$ and $q$, addressing the question raised in~\cite{arxiv24} regarding the potential for further improvement in the approximation ratio. Let $\alpha$ denote the approximation ratio for metric TSP. Now, $\alpha = 1.5 - 10^{-36}$~\cite{DBLP:conf/ipco/KarlinKG23}.
Our results are summarized as follows.

For the parameter $p$, we first introduce a simple $(\alpha + 1)$-approximation algorithm with running time $2^{\O(p)} + n^{\O(1)}$. 
This improves the previous 2.5-approximation algorithm with running time $2^{\O(p \log p)} \cdot n^{\O(1)}$~\cite{arxiv24}. Then, we propose a more refined $1.5$-approximation algorithm with running time $2^{\O(p \log p)} \cdot n^{\O(1)}$.
{Moreover, using the results from \cite{DBLP:journals/siamcomp/TraubVZ22}, we propose an $(\alpha+\varepsilon)$-approximation algorithm with running time $n^{\O(p/\varepsilon)}$ for any constant $\varepsilon>0$. 
Hence, when $p=\O(1)$, the approximation ratio almost matches that of metric TSP.}

For the parameter $q$, we propose an FPT $3$-approximation algorithm, which significantly improves the previous FPT $11$-approximation algorithm~\cite{ZhouLG22+}. Both algorithms run in $2^{\O(q \log q)} \cdot n^{\O(1)}$ time. Our result also matches the approximation ratio of an earlier algorithm with a non-FPT running time of $n^{\O(q+1)}$~\cite{ZhouLG22}.
{Similarly, we further improve this result by presenting an $(\alpha+\varepsilon)$-approximation algorithm with running time $n^{\O(q/\varepsilon)}$ for any constant $\varepsilon>0$.}

A summary of the previous and our results is provided in Table~\ref{res}. Our improvements arise from the application of novel techniques for handling non-metric graphs, which may also prove valuable for tackling other related problems. 

\begin{table}[t]
    \centering
    \begin{tabular}{c|ccc}
        \toprule
        \textbf{Parameter} & \textbf{Approximation} & \textbf{Running Time} & \textbf{Reference}\\
        \midrule
        \multirow{4}*{$p$} &  $3$ & $2^{\mathcal{O}({p\log p})}\cdot n^{\O(1)}$ & Zhou \el~\cite{ZhouLG22}\\
        
         &  $2.5$ & $2^{\mathcal{O}({p\log p})}\cdot n^{\O(1)}$ & Bampis \el~\cite{arxiv24}\\
        
        \cline{2-4}
         & $\alpha+1$ & ${2^{\mathcal{O}(p)} + n^{\O(1)}}$& \multirow{3}*{\textbf{This paper}}\\
        &  ${1.5}$ & $2^{\mathcal{O}({p\log p})}\cdot n^{\O(1)}$&\\
        
        &  $\alpha+\varepsilon$ & $n^{\O(p/\varepsilon)}$ &\\
        \midrule
        \multirow{3}*{$q$} &  $3$ & $n^{\O(q+1)}$ & Zhou \el~\cite{ZhouLG22}\\
        &  $11$ & $2^{\mathcal{O}({q\log q})}\cdot n^{\O(1)}$ & Zhou \el~\cite{ZhouLG22+}\\
        \cline{2-4}
        & ${3}$ & $2^{\mathcal{O}({q\log q})}\cdot n^{\O(1)}$& \multirow{2}*{\textbf{This paper}}\\
        &  $\alpha+\varepsilon$ & $n^{\O(q/\varepsilon)}$ &\\
        \bottomrule
    \end{tabular}
    \caption{A summary of the previous and our results on approximation algorithms for TSP parameterized by $p$ and $q$.}
    \label{res}
\end{table}

\section{Preliminary}
Let $G = (V, E, w)$ be a complete input graph with $\lvert V \rvert = n$ and edge weights $w: E \to \mathbb{R}_{\geq 0}$.  
For any subset $E' \subseteq E$, define the total weight as $w(E') \coloneqq \sum_{e \in E'} w(e)$.  
We say that $G$ is a \emph{metric} graph if, for all $a, b, c \in V$, the following conditions hold:  
$w(a, a) = 0$, $w(a, b) = w(b, a)$, and the triangle inequality $w(a, b) \leq w(a, c) + w(c, b)$.

Although the input graph contains no multiple edges, some operations (e.g., union of two graphs) may introduce multi-edges. Hence, we allow multi-edge graphs throughout the paper.
For any $V' \subseteq V$, let $G[V']$ denote the subgraph of $G$ induced by $V'$.  
For any subgraph $G' \subseteq G$, let $V(G')$ and $E(G')$ denote its vertex and edge sets, respectively.

A \emph{spanning $t$-forest} $\mathcal{F}$ in $G$ is a set of $t$ vertex-disjoint trees that together cover all vertices in $V$.  
A \emph{matching} $\mathcal{M}$ in $G$ is a set of vertex-disjoint edges such that $V(\mathcal{M}) = V$.

A triangle on three distinct vertices $a$, $b$, and $c$ is denoted by $\Delta(a, b, c)$.  
We say that a triangle is \emph{violating} if it violates the triangle inequality.  
A subset $V' \subseteq V$ is called a \emph{violating set} (VS) if removing $V'$ from $G$ results in a metric subgraph $G[V \setminus V']$.  
Let $p$ be the number of vertices involved in violating triangles, and $q$ the size of a minimum VS. 

For any (multi-)graph $G'=(V',E',w')$, the \emph{degree} of a vertex is the number of incident edges. 
Let $Odd(G')$ denote the set of odd-degree vertices. 
Then, $G'$ is \emph{Eulerian} if it is connected and $Odd(G')=\emptyset$.
A \emph{tour} $T=\overline{v_1v_2\ldots v_s}$ in $G'$ is a sequence of vertices (allowing repetitions), inducing the edge set $E'(T)\coloneqq\{v_1v_2,v_2v_3,\ldots,v_sv_1\}$, and weight $w'(T)\coloneqq w'(E'(T))$.
Moreover, 
\begin{itemize}
    \item $T$ is a \emph{simple tour} if it has no repeated vertices, and a \emph{TSP tour} if it further satisfies $\size{E'(T)}=\size{V'}$;
    \item $T$ is an \emph{Eulerian tour} if $E'(T)$ is exactly the set of all edges in $G'$. 
    We may use $G'$ to refer to its Eulerian tour.
\end{itemize}

Given a tour $T = \overline{...xyz...}$, removing $y$ yields a new tour $\overline{...xz...}$, called \emph{shortcutting} the vertex $y$ (or edges $xy$ and $yz$), or \emph{taking a shortcut} on $T$.
A simple tour $\overline{v_1v_2...v_s}$ in $G'$ defines an \emph{$s$-path} $A = v_1v_2...v_s$, consisting of $s{-}1$ edges in $E'(A) \coloneqq \{v_1v_2, v_2v_3, ..., v_{s-1}v_s\}$. A $1$-path is a single vertex, and a $\size{V'}$-path in $G'$ is called a \emph{TSP path}. A path is also referred to as a \emph{chain}.

Throughout the paper, let $T^*$ denote an optimal TSP tour in the graph $G$, and let $\OPT$ denote its weight, i.e., $\OPT=w(T^*)$.

In our algorithms, the term ``guess'' refers to enumerating all possible candidates and selecting the best found solution. For convenience, we may regard a set of edges as a graph and a set of trees as the union of their edge sets. We interpret such combinations as the graph formed by all involved edges.

\section{TSP Parameterized by $p$}\label{SC3}
We first define some notations. A vertex is \emph{bad} if it appears in a violating triangle, and \emph{good} otherwise. 
Let $V=V_b\cup V_g$, where $V_b$ (resp., $V_g$) denotes the set of bad (resp., good) vertices. Note that $\size{V_b}=p$. 
We assume $\min\{\size{V_g}, \size{V_b}\}\geq 3$, as otherwise one can easily obtain an FPT $\alpha$-approximation algorithm with running time $2^{\O(p)}+n^{\O(1)}$ (see the appendix~\ref{appa}).

By definition, we have the following property.
\begin{property}\label{LB0}
Every triangle containing a good vertex satisfies the triangle inequality. 
\end{property}

\subsection{The First Algorithm}\label{SC3.1}
In this subsection, we propose an FPT $(\alpha+1)$-approximation algorithm (ALG.1) with running time $2^{\O(p)} + n^{\O(1)}$.

Our ALG.1 is simple. First, it arbitrarily selects a good vertex $o$. Then, it computes an optimal TSP tour $T_b$ in $G[V_b\cup\{o\}]$ using the Dynamic Programming (DP) method~\cite{bellman1962dynamic,held1962dynamic}. Next, it computes an $\alpha$-approximate TSP tour $T_g$ in $G[V_g]$ by applying an $\alpha$-approximation algorithm for metric TSP. Finally, it constructs a TSP tour $T_1$ in $G$ by taking a shortcut on $T_b\cup T_g$.

The TSP tour $T_b$ in ALG.1 contains exactly one good vertex. This trick enables us to take a shortcut on $T_b\cup T_g$ to obtain a TSP tour in $G$ without increasing the weight.

Our ALG.1 is described in Algorithm~\ref{alg1}.

\begin{algorithm}[t]
\caption{ALG.1}
\label{alg1}
\textbf{Input:} An instance $G=(V=V_g\cup V_b, E, w)$. \\
\textbf{Output:} A feasible solution.
\begin{algorithmic}[1]
\State\label{1.1} Arbitrarily select a good vertex $o\in V_g$.

\State\label{1.2} Compute an optimal TSP tour $T_b$ in $G[V_b\cup\{o\}]$ using the DP method~\cite{bellman1962dynamic}.

\State\label{1.3} Compute an $\alpha$-approximate TSP tour $T_g$ in $G[V_g]$ by applying an $\alpha$-approximation algorithm for metric TSP.

\State\label{1.4} Take a shortcut on $T_b\cup T_g$ to obtain a TSP tour $T_1$ in $G$.
\end{algorithmic}
\end{algorithm}

\begin{lemma}\label{LB1}
It holds that $\OPT\geq w(T^*_b)$ and $\OPT\geq w(T^*_g)$, where $T^*_b$ (resp., $T^*_g$) denotes an optimal TSP tour in $G[V_b\cup\{o\}]$ (resp., $G[V_g]$).
\end{lemma}
\begin{proof}
First, we prove $\OPT\geq w(T^*_b)$.
Given the optimal TSP tour $T^*$ in $G$, by Property~\ref{LB0}, we can directly shortcut all good vertices in $V_g\setminus\{o\}$ to obtain a TSP tour $T'$ in $G[V_b\cup\{o\}]$ with $\OPT\geq w(T')$. Since $T^*_b$ is an optimal TSP tour in $G[V_b\cup\{o\}]$, we have $\OPT\geq w(T')\geq w(T^*_b)$.

Next, we prove $\OPT\geq w(T^*_g)$.
Recall that we assume $\min\{\size{V_g}, \size{V_b}\}\geq 3$. Then there must exist a bad vertex adjacent to a good vertex in $T^*$. By Property~\ref{LB0}, we can shortcut the bad vertex via its adjacent good vertex without increasing the weight. Repeating this process for all bad vertices yields a tour $T''$ in $G[V_g]$ such that $\OPT \geq w(T'')$. Similarly, we have $\OPT\geq w(T'')\geq w(T^*_g)$.
\end{proof}

\begin{theorem}\label{t1}
For TSP parameterized by $p$, ALG.1 achieves an FPT $(\alpha+1)$-approximation with running time $2^{\O(p)}+n^{\O(1)}$.
\end{theorem}
\begin{proof}
First, we analyze the solution quality of ALG.1.

Since $T_1$ is obtained by taking a shortcut on $T_b\cup T_g$, we assume that $T_1=\overline{ou_1...u_pv_1...v_{n-p-1}}$, where $T_b=\overline{ou_1...u_p}$ and $T_g=\overline{ov_1...v_{n-p-1}}$. Since $o$ is a good vertex, the triangle $\Delta(o,v_1,u_p)$ satisfies the triangle inequality by Property~\ref{LB0}, which implies $w(u_p,o)+w(o,v_1)\geq w(u_p,v_1)$. Thus, we have $w(T_b)+w(T_g)\geq w(T_1)$. Since $T_b$ (resp., $T_g$) is an optimal (resp., $\alpha$-approximate) TSP tour in $G[V_b\cup\{o\}]$ (resp., $G[V_g]$), by Lemma~\ref{LB1}, we have $(\alpha+1)\cdot\OPT\geq w(T_1)$.

Next, we analyze the running time of ALG.1. 

Steps~\ref{1.1}, \ref{1.3}, and \ref{1.4} take polynomial time, i.e., $n^{\O(1)}$, while Step~\ref{1.2} is dominated by the computation of an optimal TSP tour. 
Since $\size{V_b}+1=p+1$, the optimal TSP tour can be computed in $2^{\O(p)}$ time using the DP method~\cite{bellman1962dynamic}. Thus, the overall running time is $2^{\O(p)}+n^{\O(1)}$.
\end{proof}

By Theorem~\ref{t1}, using the $(1.5-10^{-36})$-approximation algorithm~\cite{DBLP:conf/stoc/KarlinKG21} for metric TSP, ALG.1 achieves an FPT approximation ratio of $2.5-10^{-36}$, improving the algorithm in~\cite{arxiv24} in both approximation ratio and running time.

\subsection{The Second Algorithm}\label{SC3.2}
In this subsection, we propose an FPT $1.5$-approximation algorithm (ALG.2) with running time $2^{\O(p\log p)}\cdot n^{\O(1)}$.

ALG.2 proceeds as follows. 
To handle bad vertices, it first guesses the subgraph $T^*[V_b]$, i.e., the subgraph of the optimal TSP tour $T^*$ induced by $V_b$. This forms a set of paths on bad vertices, called \emph{bad chains} $\A$. 

ALG.2 then augments $\A$ into a TSP tour in $G$.
It first constructs a minimum-weight \emph{constrained spanning tree (CST)} $F_\A$ in $G$, i.e., a tree such that $E(\A)\subseteq E(F_\A)$, each $i$-degree vertex in $\A$ is incident to exactly $i$ bad vertices in $F_\A$, and each $2$-degree vertex in $\A$ remains of degree 2 in $F_\A$.
It then adds a set of edges $\M_\A$ to fix the parity of $Odd(F_\A)$ and computes a TSP tour $T_2$ in $G$ by shortcutting on $F_\A \cup M_\A$. 

Algorithm~\ref{alg2} gives full details, with an example in Figure~\ref{fig1}.\footnote{Strictly speaking, Steps~\ref{2.2}-\ref{2.7} are conceptually nested within Step~\ref{2.1}, as they depend on the guessed set $\A$ in Step~\ref{2.1}. That is, the algorithm enumerates all possible choices of $\A$, and for each choice, it executes Steps~\ref{2.2}-\ref{2.7}. However, for clarity and conciseness, we present the algorithm in a sequential manner.}


\begin{algorithm}[t]\label{tree+matching}
\caption{ALG.2}
\label{alg2}
\textbf{Input:} An instance $G=(V=V_g\cup V_b, E, w)$. \\
\textbf{Output:} A feasible solution.
\begin{algorithmic}[1]
\State\label{2.1} Guess the set of bad chains $\A$ in $T^*$.

\State\label{2.2} Construct a complete graph $\widetilde{G}=(V_\A\cup V_g,\widetilde{E},\widetilde{w})$ by contracting each $A\in\A$ into a vertex $u_A$, where $V_\A$ is the set of contracted vertices, and $\widetilde{w}$ is defined as follows:
(1) $\widetilde{w}(u_A,v)=\min\{w(u,v),w(u',v)\}$ for any $A=u...u'\in\A$ and $v\in V_g$; (2) $\widetilde{w}(v,v')=w(v,v')$ for any $v,v'\in V_g$; (3) $\widetilde{w}(u_A,u_{A'})=+\infty$ for any $A, A'\in\A$.

\State\label{2.3} Find a minimum-weight spanning tree $\widetilde{F}$ in $\widetilde{G}$, and obtain the corresponding set of edge $F$ in $G$ with $w(F)=\widetilde{w}(\widetilde{F})$. 

\State\label{2.4} Let $F_\A=F\cup \A$, which forms a CST in $G$.

\State\label{2.5} Construct a complete graph $G'=(V'=Odd(F_\A), E',w')$: For any $u,u'\in V'$, set $w'(u,u')=w(A)$ if $u...u'=A\in\A$, and $w'(u,u')=w(u, u')$ otherwise.

\State\label{2.6} Find a minimum-weight matching $\M'_\A$ in $G'$, and obtain the corresponding set of edges $\M_\A$ in $G$ with $w(M_\A)=w'(M'_\A)$.

\State\label{2.7} Obtain a TSP tour $T_2$ in $G$ (see Lemma~\ref{LB2.3}).
\end{algorithmic}
\end{algorithm}

\begin{figure}[t]
    \centering
    \resizebox{0.88\linewidth}{!}{
    \begin{subfigure}[b]{0.495\linewidth}
        \centering
        \includegraphics[width=0.8\linewidth]{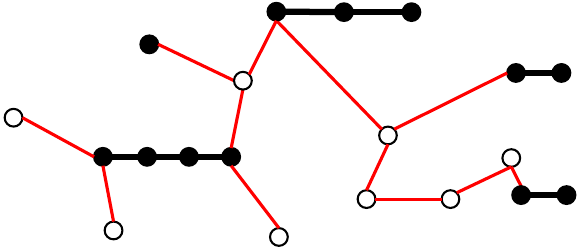}
        \caption{$F_\A=F\cup\A$.}
    \end{subfigure}
    \hfill
    \begin{subfigure}[b]{0.495\linewidth}
        \centering
        \includegraphics[width=0.8\linewidth]{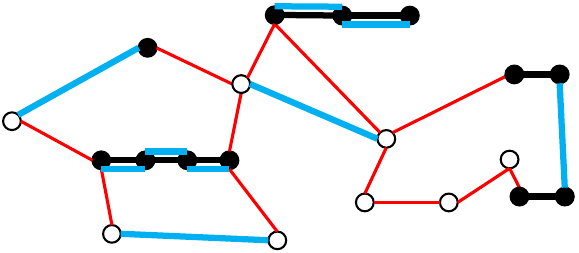}
        \caption{$F_\A\cup M_\A$.}
    \end{subfigure}
    }

    \vspace{1mm}

    \resizebox{0.88\linewidth}{!}{
    \begin{subfigure}[b]{0.495\linewidth}
        \centering
        \includegraphics[width=0.85\linewidth]{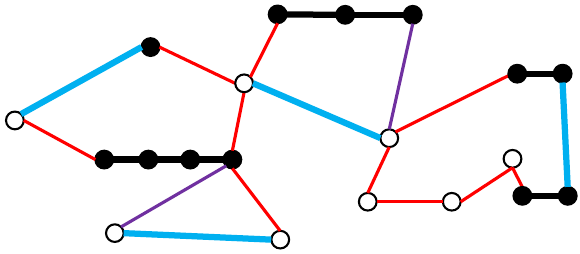}
        \caption{$T'_\A$.}
    \end{subfigure}
    \hfill
    \begin{subfigure}[b]{0.495\linewidth}
        \centering
        \includegraphics[width=0.85\linewidth]{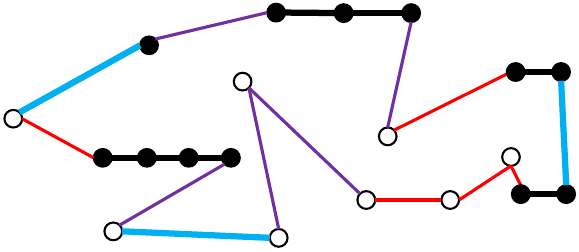}
        \caption{$T_2$.}
    \end{subfigure}
    }
    \caption{An illustration of our ALG.2, where each black (resp., white) node denotes a bad (resp., good) vertex. 
    In (a), the edges in $E(F)$ (resp., $E(\A)$) are shown in red (resp., black); In (b), the edges in $\M_\A$ are shown in blue; In (c), $\T'_\A$ is obtained by shortcutting repeated edges on the paths in $\A\cap\A'$ (see Lemma~\ref{LB2.3}), with the new edges shown in purple; In (d), the TSP tour $T_2$ is obtained by taking shortcuts on $\T'_\A$ (see Lemma~\ref{LB2.3}), with the new edges also shown in purple.}
    \label{fig1}
\end{figure}

We first compare our ALG.2 with previous algorithms.

The previous FPT 3-approximation algorithm~\cite{ZhouLG22} and $2.5$-approximation algorithm~\cite{arxiv24} follow a similar framework: they guess the set of bad chains $\A$, extract a set $V'_b$ of endpoint vertices from $\A$, and then find a minimum-weight spanning $\size{\A}$-forest $\F$ in $G[V'_b \cup V_g]$~\cite{khachay2016approximability}, where each vertex $b_i \in V'_b$ lies in a different tree. By adding a set of edges $E'$, they then form a TSP tour on $V_b$. 
The algorithm in~\cite{ZhouLG22} takes shortcuts on $\A \cup E' \cup 2\F$ for a $3$-approximation, while the algorithm in~\cite{arxiv24} uses a matching $\M$ on $Odd(\A \cup E' \cup \F)$ to shortcut $\A \cup E' \cup \F\cup\M$ to achieve a $2.5$-approximation.

In contrast, ALG.2 computes a CST $F_\A$ in $G$, avoiding the need for $E'$. However, $G[Odd(F_\A)]$ may be non-metric, so we cannot simply use a minimum-weight matching~\cite{lawler1976combinatorial}. 
We build an auxiliary graph $G'$ (Step~\ref{2.5}) and find a matching $\M'_\A$ in $G'$, then shortcut on $F_\A \cup M_\A$ to obtain a TSP tour, where $\M_\A$ is the corresponding set of edges in $G$.

Hence, ALG.2 introduces key differences and new ideas, requiring refined analysis for the final approximation ratio.

\begin{lemma}\label{LB2.1}
ALG.2 computes a CST $F_\A$ with $w(F_\A)\leq\OPT$, and a set of edges $\M_\A$ with $w(M_\A)\leq\frac{1}{2}\cdot\OPT$.
\end{lemma}
\begin{proof}
We begin by proving that ALG.1 computes a valid CST $F_\A$ with $w(F_\A)\leq\OPT$.

Since $\A$ is a set of paths (and thus a set of trees), we regard $\A$ as a graph.
For each $i\in\{0,1,2\}$, denote by $V^i_b$ the set of bad vertices of degree $i$ in the graph $\A$.

First, we show that $F_\A$ is a CST in $G$.
By Step \ref{2.2} of ALG.2, the tree $F$ in $G$, which corresponds to the tree $\widetilde{F}$ in $\widetilde{G}$, consists only of edges between good vertices and edges between one good vertex and one bad vertex in $V^0_b\cup V^1_b$. Moreover, since $\widetilde{F}$ is a spanning tree in $\widetilde{G}$, the tree $F_\A=F\cup\A$ is a valid CST in $G$ by definition.

Then, we show that $F_\A$ is a minimum-weight CST in $G$.
Consider any minimum-weight CST $F'$ in $G$. By definition, $E(\A)\subseteq E(F')$, and the set $E(F')\setminus E(\A)$ consists only of edges between good vertices and edges between one good vertex and one bad vertex in $V^0_b\cup V^1_b$. Thus, by Step \ref{2.2} of ALG.2, $E(F')\setminus E(\A)$ corresponds to a spanning tree $\widetilde{F}'$ in $\widetilde{G}$ with 
\[
\widetilde{w}(\widetilde{F}')\leq w(F')-w(\A).
\]
Since $F_\A$ satisfies $w(F_\A)=\widetilde{w}(\widetilde{F})+w(\A)$ and $\widetilde{F}$ is a minimum-weight spanning tree in $\widetilde{G}$, we have 
\begin{align*}
w(F_\A)&=\widetilde{w}(\widetilde{F})+w(\A)\leq \widetilde{w}(\widetilde{F}') + w(\A)\leq w(F').   
\end{align*}
Thus, $F_\A$ is a minimum-weight CST in $G$.



Finally, by arbitrarily deleting an edge between one good vertex and one bad vertex in $E(T^*)$, we obtain a CST in $G$ with weight at most $\OPT$. 
Thus, we have $w(F_\A)\leq \OPT$.

Now, we prove that ALG.1 computes a matching $\M_\A$ with $w(M_\A)\leq\frac{1}{2}\cdot\OPT$.

By the definition of $\F_\A$, each $2$-degree vertex in the graph $\A$ remains a $2$-degree vertex in $F_\A$. Thus, 
\[
Odd(F_\A)\subseteq V^0_b\cup V^1_b\cup V_g.
\]

Next, we try to shortcut vertices not in $Odd(F_\A)$ on the optimal TSP tour $T^*$ with a non-increasing weight. 

First, we consider vertices in $V_b$.
For any bad chain $u...u'\in\A$, if at least one terminal is not in $Odd(F_\A)$, say $u$, there must exist a good vertex $v$ in $T^*$ that is incident to $u$, and then we can use $v$ to repeatedly shortcut all vertices on $u...u'$ that are not in $Odd(F_\A)$ without increasing the weight by Property~\ref{LB0}; if both terminals are in $Odd(F_\A)$, we do not shortcut the vertices on $u...u'$.

Then, we consider vertices in $V_g$.
By Property~\ref{LB0}, we can directly shortcut all vertices in $V_g$ that are not in $Odd(F_\A)$ without increasing the weight. 

Therefore, we can obtain a simple tour $T$ with $w(T)\leq\OPT$, which consists of all vertices in $Odd(F_\A)$ as well as all bad vertices in $V^2_b$ that are part of bad chains where both terminals belong to $Odd(F_\A)$. 

By Step \ref{2.5} of ALG.2, the tour $T$ corresponds to a TSP tour $T'$ in $G'$, where $w(T)=w'(T')$. 
Clearly, $\size{Odd(F_\A)}$ is even.
Thus, $T'$ can be decomposed into two edge-disjoint matchings in $G'$. 
Since $\M'_\A$ is a minimum-weight matching in $G'$, we must have $w'(M'_\A)\leq\frac{1}{2}\cdot w'(T')=\frac{1}{2}\cdot w(T)\leq\frac{1}{2}\cdot\OPT$.

Finally, since $\M_\A$ is the set of edges in $G$ corresponding to $M'_\A$, by Step \ref{2.6} of ALG.2, we have $w(M_\A)=w'(\M'_\A)\leq\frac{1}{2}\cdot\OPT$.
\end{proof}

\begin{lemma}\label{LB2.3}
In Step~\ref{2.6} of ALG.2, a TSP tour $T_2$ in $G$ can be found in $\O(n)$ time with weight at most $\frac{3}{2}\cdot\OPT$. 
\end{lemma}
\begin{proof}
Clearly, $F_\A\cup M_\A$ forms an Eulerian graph $G_\A$ with $V(G_\A)=V$ and $\O(n)$ edges. Hence, an Eulerian tour $T_\A$ can be computed in $\O(n)$ time~\cite{cormen2022introduction}. 

Note that $w(T_\A)=w(F_\A)+w(M_\A)\leq\frac{3}{2}\cdot\OPT$. To prove the lemma, it suffices to show how to compute in $\O(n)$ time a TSP tour $T_2$ in $G$ with weight at most $T_\A$.

We first observe that 
\begin{enumerate}
    \item[(1)] the edges between bad vertices in $E(F_\A)$ are exactly those in $E(\A)$;
    \item[(2)] the edges between bad vertices in $E(M_\A)$ form a set of vertex-disjoint paths $\A'$ since $\M'_\A$ is a matching in $G'$.
\end{enumerate}
By Step~\ref{2.5} of ALG.2, we further know that 
\begin{enumerate}
    \item[(3)] for each path $A\in \A'$, either $A\in\A$, or $A$ is edge-disjoint from all paths in $\A$.
\end{enumerate}
Thus, the edges on the paths in $\A\cap\A'$ (resp., $\A\setminus\A'$) appear exactly twice (resp., once) in $T_\A$ (see (b) in Figure~\ref{fig1}).

Next, we shortcut the repeated edges on the paths in $\A\cap\A'$.
Consider any path $ab...b'a'\in\A\cap\A'$. 
By (1), (2), and (3), in $V_b$, there is only one vertex $b$ (resp., $b'$) adjacent to $a$ (resp., $a'$). Since $G_\A$ is also Eulerian, $G_\A$ is connected, and both $a$ and $a'$ have even degree. Hence, either $a$ or $a'$, say $a$, must be adjacent to at least two good vertices. 
Thus, we can use a good neighbor of $a$ to shortcut the repeated edges on the path $ab...b'a'$ without increasing the weight by Property~\ref{LB0}.

Denote the new Eulerian tour by $T'_\A$.
We know that
\begin{enumerate}
    \item[(4)] each edge between bad vertices appears once, and each bad vertex is incident to at most two bad vertices (see (c) in Figure~\ref{fig1}).
\end{enumerate}
 
Then, we shortcut the repeated bad vertices in $T'_\A$.
Suppose a bad vertex $a$ appears at least twice in $T'_\A$. 
By (4), each appearance of $a$, except for one, must be adjacent to at least one good vertex. 
Thus, each of these can be shortcut via a good neighbor without increasing the weight by Property~\ref{LB0}.

Last, repeated good vertices can clearly be shortcut without increasing the weight. Hence, given $\T_\A$, a TSP tour $T_2$ in $G$ can be found with weight at most $w(\T_\A)$ in $\O(n)$ time.
\end{proof}

\begin{theorem}\label{t2}
For TSP parameterized by $p$, ALG.2 achieves an FPT $1.5$-approximation ratio with running time $2^{\O(p\log p)}\cdot n^{\O(1)}$.
\end{theorem}
\begin{proof}
For the running time, there are $\O(2^p\cdot p!)\subseteq 2^{\O(p\log p)}$ possible structures for the set of bad chains $\A$~\cite{arxiv24}. Hence, Step~\ref{2.1} involves $2^{\O(p\log p)}$ guesses. For each guess, Steps \ref{2.2}-\ref{2.7} take $\O(n^3)$ time, dominated by the computation of a minimum-weight matching~\cite{lawler1976combinatorial}. 
Thus, the overall running time is $2^{\O(p\log p)}\cdot n^{\O(1)}$.
For the solution quality, by Lemma \ref{LB2.3}, ALG.2 achieves an approximation ratio of $1.5$.
This completes the proof.
\end{proof}

\subsection{The Third Algorithm}
In this subsection, we propose an $(\alpha+\varepsilon)$-approximation algorithm (ALG.3) with running time $n^{\O(p/\varepsilon)}$ for any constant $\varepsilon>0$.

ALG.3 is based on an $(\alpha+\varepsilon)$-approximation algorithm for the metric $k$-TSP path ($k$-TSPP): given a metric graph with $n$ vertices and $k$ vertex pairs $(s_1,t_1),...,(s_k,t_k)$, the objective is to find a path between $s_i$ and $t_i$ for each $i$ such that the total weight of these $k$ paths is minimized and every vertex of the graph is covered by at least one path. For metric $k$-TSPP, a 3-approximation algorithm is well-known, while B{\"{o}}hm \el~\cite{DBLP:conf/soda/0001FMS25} proposed an algorithm with an approximation ratio of $1+2\cdot e^{-1/2}<2.214$. Moreover, Traub \el~\cite{DBLP:journals/siamcomp/TraubVZ22} proposed an $(\alpha+\varepsilon)$-approximation algorithm for $\Phi$-TSP with bounded \emph{interface size}. If the interface size is $s$, the running time of their algorithm is $n^{\O(s/\varepsilon)}$.
It can be verified that metric $k$-TSPP is a special case of $\Phi$-TSP with an {interface size} of at most $2k$. (The interface size is $2k$ when $s_i\neq t_j$ for all $1\leq i\leq j\leq k$.)
Hence, metric $k$-TSPP admits an $(\alpha+\varepsilon)$-approximation algorithm with running time $n^{\O(k/\varepsilon)}$.

Our ALG.3 simply reduces TSP parameterized by $p$ to metric TSPP. 

Suppose that $T^*=\overline{s_1a_1...b_1t_1...s_2a_2...b_2t_2...s_\ell a_\ell...b_\ell t_\ell...}$, where $V'_g=\{s_i,t_i\mid 1\leq i\leq\ell\}$ are the good vertices that are incident to bad vertices, and $\A=\{a_i...t_i\mid 1\leq i\leq\ell\}$ is the set of bad chains that ALG.3 attempts to guess in the previous subsection.

ALG.3 first guesses the path set $\P=\{P_1,...,P_\ell\}$, where  $P_i=s_ia_i...b_it_i$. Then, it constructs an $\ell$-TSPP instance $G'$ with $V_g$ being the vertex set and $\ell$ vertex pairs $(t_1,s_2),(t_2,s_3),...,(t_\ell,s_1)$, and calls the $(\alpha+\varepsilon)$-approximation algorithm~\cite{DBLP:journals/siamcomp/TraubVZ22} to compute $\ell$ paths $\P'$. By the definition of $\ell$-TSPP, patching all paths in $\P\cup\P'$ yields a feasible solution to TSP.

Clearly, if the path set $\P$ is guessed correctly, ALG.3 would compute a solution to TSP with weight at most $w(\P)+w(\P')\leq w(\P)+(\alpha+\varepsilon)\cdot(w(T^*)-w(\P))\leq (\alpha+\varepsilon)\cdot\OPT$.
It can be verified that the path set $\P$ can be obtained by enumerating $2^{\O(p\log p)}\cdot n^{2p}$ possible structures. 
Moreover, since $\ell\leq p$, the $(\alpha+\varepsilon)$-approximation algorithm~\cite{DBLP:journals/siamcomp/TraubVZ22} takes $n^{\O(p/\varepsilon)}$ time. 

Therefore, we obtain the following result.

\begin{theorem}\label{t3+}
For TSP parameterized by $p$, there exists an $(\alpha+\varepsilon)$-approximation algorithm with running time $n^{\O(p/\varepsilon)}$.
\end{theorem}

\section{TSP Parameterized by $q$}\label{SC4}
In this section, we consider the parameter $q$.
It is NP-hard to compute $q$. 
In our algorithm, we first compute a minimum violating set (VS) $V_b$ in $\O(3^qn^3)$ time by using the algorithm in ~\cite{ZhouLG22}. 
 Let $V_g=V\setminus V_b$.
The vertices in $V_b$ are called \emph{bad}, and the vertices in $V_g$ are called \emph{good}. Now $\size{V_b} = q$.
We assume that $\min\{\size{V_g}, \size{V_b}\}\geq 1$. Next, we first define some notations.

In $T^*$, the edges between bad vertices form the set of \emph{bad chains} $\A$; the edges between a bad and a good vertex form the set $\B$, where each such edge is a \emph{limb}; and the edges between good vertices form the set of \emph{good chains} $\R$. 
Note that 
\begin{equation}\label{eq3.0}
\OPT=w(\A)+w(\B)+w(\R)\ \ \mbox{and}\ \ \size{\A}=\frac{1}{2}\size{\B}=\size{\R}.
\end{equation}
Let $V^i_g$ denote the set of vertices in $V_g$ that are incident to $i$ bad vertices in $T^*$, where $i\in\{0,1,2\}$. The vertices in $V^0_g$ (resp., $V^1_g\cup V^2_g$) are called \emph{internal vertices} (resp., \emph{anchors}). An anchor is a \emph{single anchor} if it is in $V^2_g$, and a \emph{pair anchor} otherwise. Let $V_a=V^1_g\cup V^2_g$ be the set of all anchors. Then, we have
\begin{equation}\label{eq3.1}
\size{\A}=\frac{1}{2}\size{\B}=\size{\R}\leq q\quad\mbox{and}\quad\size{V_a}\leq 2q.
\end{equation}
An illustration of these notations can be found in Figure~\ref{notations}. 
Note that part of the terminology, e.g., \emph{anchor} and \emph{limb}, is adopted from \cite{DBLP:conf/stoc/BlauthN23}.

\begin{figure}
    \centering
    \includegraphics[scale=0.85]{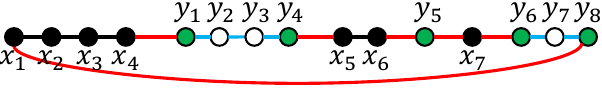}
    \caption{An illustration of our notations: the cycle denotes $T^*$; the black, white, and green nodes denote bad vertices, internal vertices, and anchors, respectively. We have $\A=\{x_1...x_4,x_5x_6,x_7\}$, $\B=\{x_4y_1,...,y_8x_1\}$, and $\R=\{y_1...y_4,y_5,y_6y_7y_8\}$. Note that $x_7\in\A$ is a 1-path.}
\label{notations}
\end{figure}

First, by directly extending ALG.3 to TSP parameterized by $q$, we have the following result.

\begin{theorem}\label{t3++}
For TSP parameterized by $q$, there exists an $(\alpha+\varepsilon)$-approximation algorithm with running time $n^{\O(q/\varepsilon)}$.
\end{theorem}

Next, we give an FPT $3$-approximation algorithm (ALG.4) with running time $2^{\O(q\log q)}\cdot n^{\O(1)}$. 

\subsection{The Algorithm}\label{SC4.1}
The high-level idea of ALG.4 is as follows.
Since a triangle containing one bad vertex and two good vertices may violate the triangle inequality, Property~\ref{LB0} in the previous section does not hold. Thus, only guessing $\A$ may not be enough to obtain a constant FPT approximation ratio.
Since $\size{V_g}$ can also be arbitrarily larger than $q$, we cannot guess $\B$ directly.
To address these issues, we use an idea in \cite{ZhouLG22+} to guess a set of limbs $\B'$ in FPT time by finding a minimum-weight spanning $\size{\A}$-forest $\F$ in $G[V_g]$ and then guessing a set of anchors in a subset of $V(\F)$, where $\B'$ is not necessarily equals to $\B$ but it satisfies $w(\B')\leq w(\B)$. Clearly, 
\begin{equation}\label{eq3.2}
\size{\F}\leq q\quad\mbox{and}\quad w(\F)\leq w(\R). 
\end{equation}

Then, we augment $\A\cup\B'\cup\F$ into an Eulerian graph. The set of good chains $\R$ ensures the existence of a set of edges $\R'$ with $w(\R')\leq w(\R)$ such that $G_\A=\A\cup\B'\cup\R'\cup\F$ is a connected graph and $\size{Odd(G_\A)\cap V(F)}$ is even for each $F\in\F$. Thus, we can find a minimum-weight matching $\M_F$ in $G[Odd(G_\A)\cap V(F)]$ for each $F\in \F$ to fix the parity of the odd-degree vertices in $G_\A$. 
We will also prove $w(M_F)\leq w(F)$. 
Let $\M=\bigcup_{F\in\F} M_F$. Finally, we obtain an Eulerian graph $G'_\A=\A\cup\B'\cup\R'\cup\F\cup\M$, and then we take shortcuts on $G'_\A$ to obtain a TSP tour without increasing the weight. 
Thus, by (\ref{eq3.0}), the tour has weight at most 
\begin{align*}
w(G'_\A)&\leq w(\A)+w(\B')+w(\R')+w(\F)+w(\M)\leq w(\A)+w(\B)+w(\R)+2w(\R)\leq 3\cdot\OPT    
\end{align*} 

ALG.4 is described in Algorithm~\ref{alg3}, where {LIMB}, {CONNECT}, and {SHORTCUT} are sub-algorithms explained later. An illustration of our ALG.4 can be found in Figure~\ref{fig3}.

\begin{algorithm}[t]
\caption{ALG.4}
\label{alg3}
\textbf{Input:} An instance $G=(V=V_g\cup V_b, E, w)$. \\
\textbf{Output:} A feasible solution.
\begin{algorithmic}[1]
\State\label{3.1} Guess the set of bad chains $\A$ in $T^*$.

\State\label{3.2} Find a minimum-weight spanning $\size{\A}$-forest $\F$ in $G[V_g]$.

\State\label{3.3} Guess a set of limbs $\B'$ by calling {LIMB}.

\State\label{3.4} Guess a set of edges $\R'$ by calling {CONNECT}.


\State\label{3.5} Find a set of edges $\M=\bigcup_{F\in\F} M_F$, where $\M_F$ is a minimum-weight matching in $G[Odd(G_\A)\cap V(F)]$ and $G_\A=\A\cup\B'\cup\R'\cup\F$ for each $F\in\F$.

\State\label{3.6} Obtain a TSP tour $T_3$ in $G$ using {SHORTCUT}.
\end{algorithmic}
\end{algorithm}

\begin{figure}[t]
    \centering
    \resizebox{0.88\linewidth}{!}{
    \begin{subfigure}[b]{0.4\linewidth}
        \centering
        \includegraphics[width=0.8\linewidth]{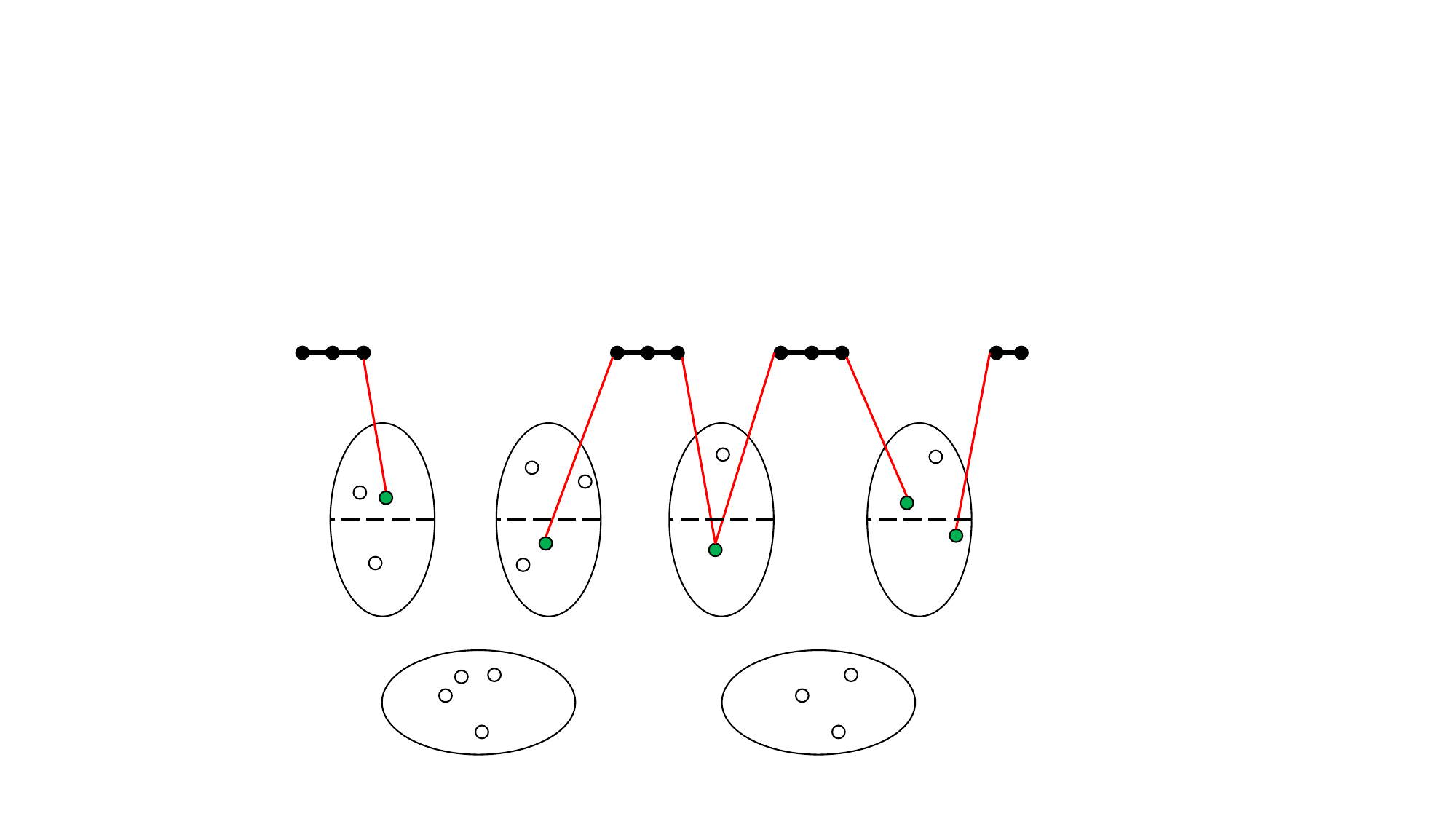}
        \caption{The $\A$ and $\B$.}
    \end{subfigure}
    \hfill
    \begin{subfigure}[b]{0.4\linewidth}
        \centering
        \includegraphics[width=0.8\linewidth]{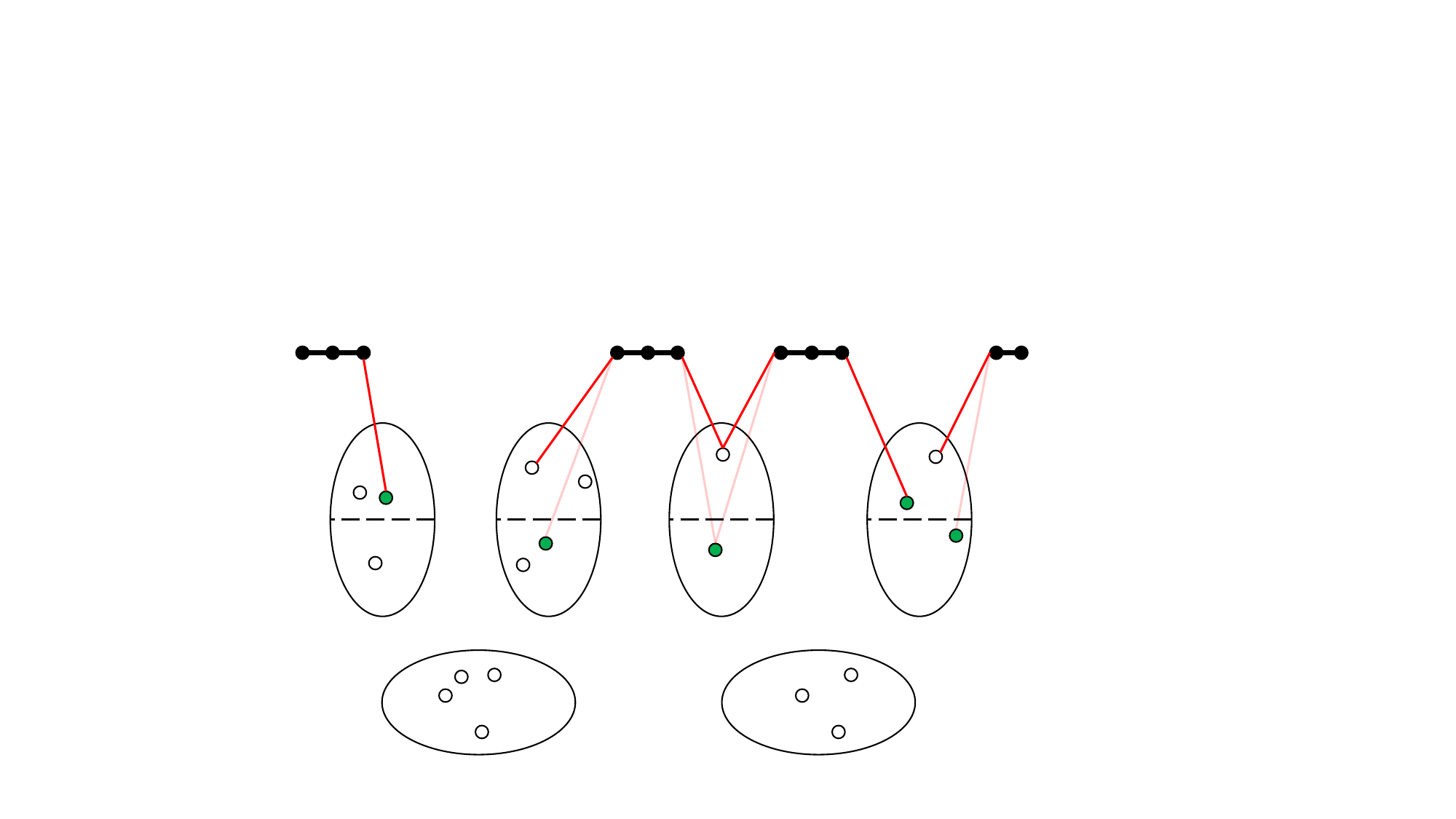}
        \caption{The guessed limbs in $\B'$.}
    \end{subfigure}
    }

    \vspace{1mm}
    
    \resizebox{0.88\linewidth}{!}{
    \begin{subfigure}[b]{0.4\linewidth}
        \centering
        \includegraphics[width=0.8\linewidth]{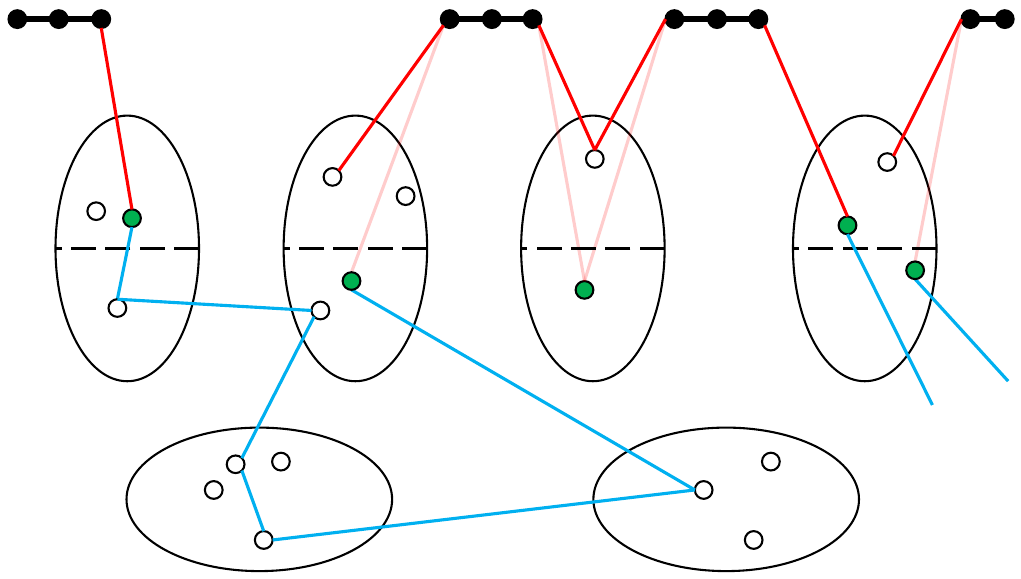}
        \caption{The chains in $\R$.}
    \end{subfigure}
    \hfill
    \begin{subfigure}[b]{0.4\linewidth}
        \centering
        \includegraphics[width=0.8\linewidth]{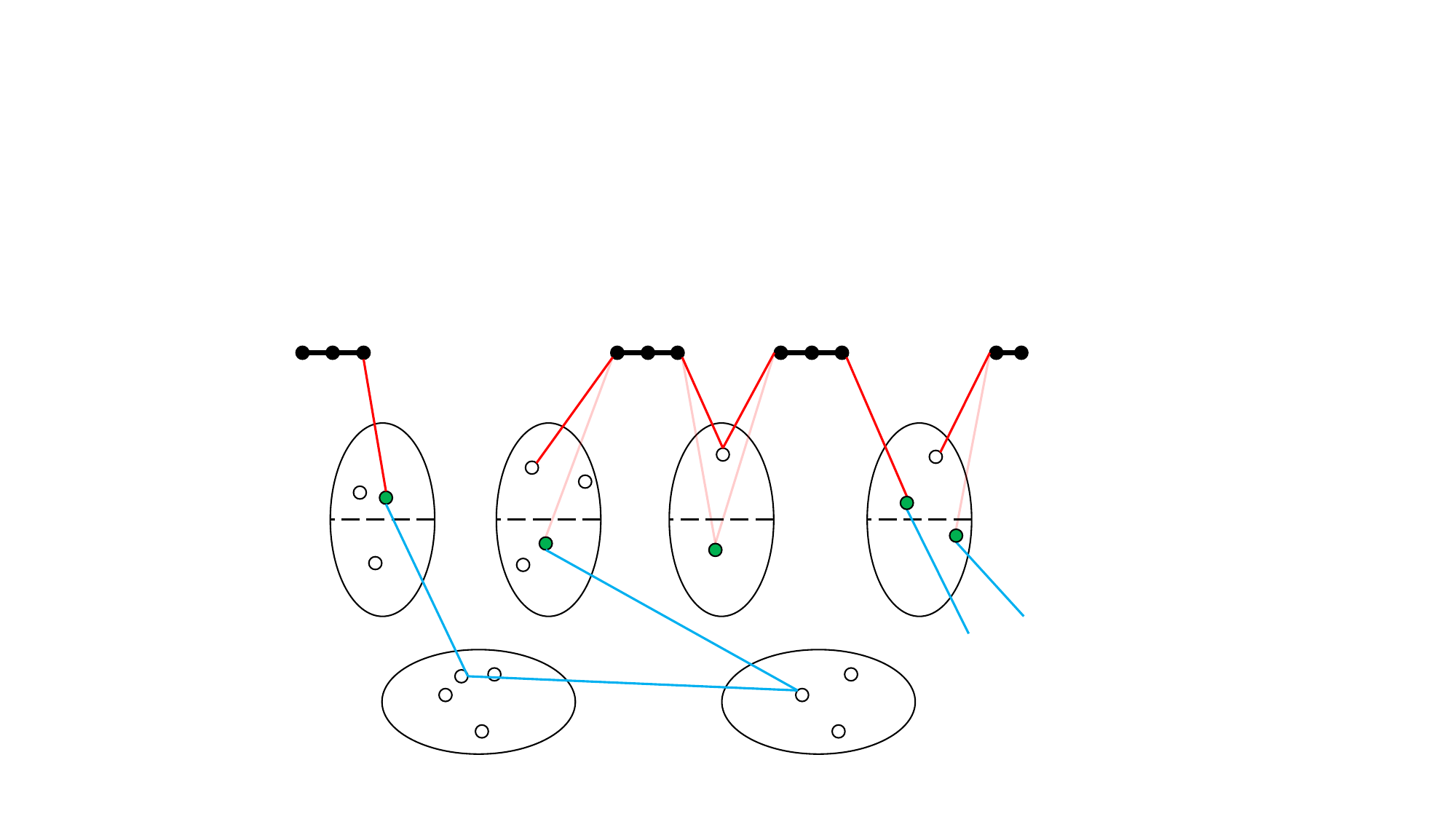}
        \caption{The chains in $\R^*$.}
    \end{subfigure}
    }

    \vspace{1mm}
    
    \resizebox{0.88\linewidth}{!}{
    \begin{subfigure}[b]{0.4\linewidth}
        \centering
        \includegraphics[width=0.8\linewidth]{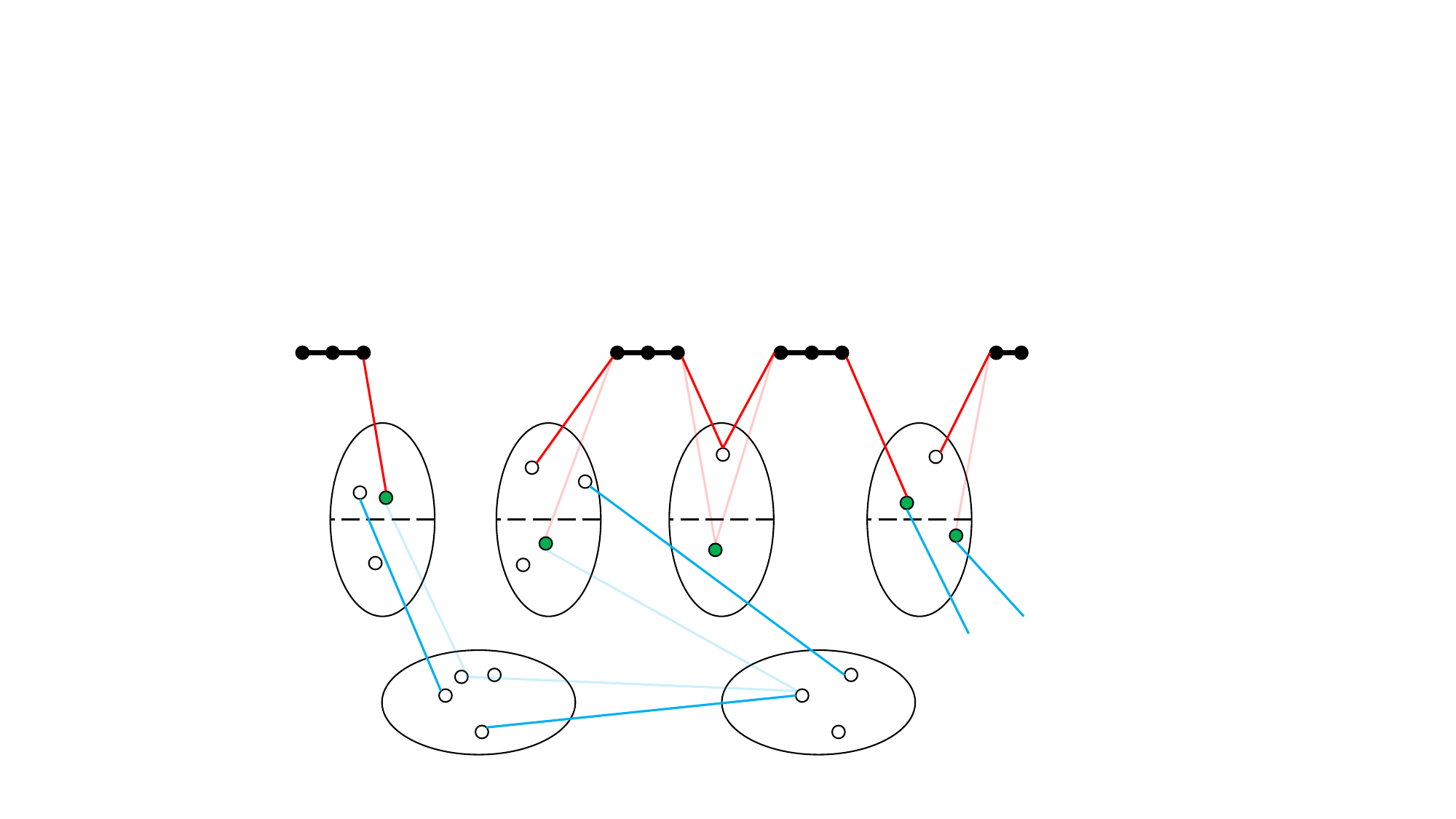}
        \caption{The guessed edges in $\R'$.}
    \end{subfigure}
    \hfill
    \begin{subfigure}[b]{0.4\linewidth}
        \centering
        \includegraphics[width=0.8\linewidth]{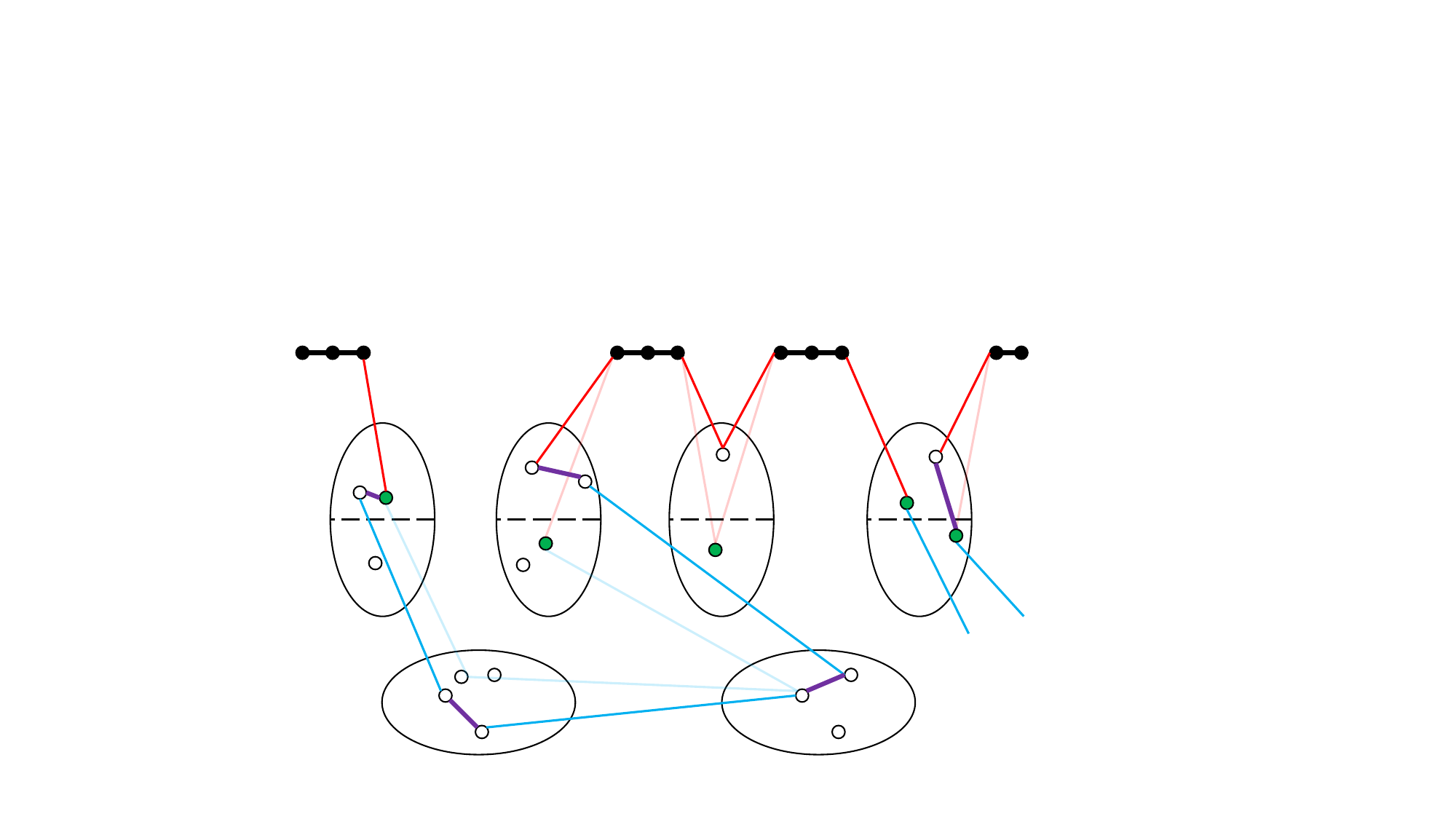}
        \caption{The edges in $\M$.}
    \end{subfigure}
    }
    \caption{An illustration of our ALG.4. Black (resp., white) nodes denote bad (resp., internal) vertices; green nodes are anchors. Each ellipse denotes a tree $F\in\F$, with the upper part indicating its potential set $V_F$.
    In (a), the edges in $E(\A)$ (resp., limbs in $\B$) are shown in black (resp., red); 
    In (b), the guessed limbs in $\B'$ are shown in red, and the white and green nodes incident to red edges denote guessed anchors; 
    In (c), the edges in $E(\R)$ are shown in blue; 
    In (d), the edges in $E(\R^*)$ are shown in blue;
    In (e), the guessed edges in $\R'$ are shown in blue;
    In (f), the edges in $\M$ are shown in purple.
    }
    \label{fig3}
\end{figure}

We remark that, instead of guessing a set of limbs $\B'$, the previous FPT $11$-approximation algorithm~\cite{ZhouLG22+} mainly uses the spanning $\size{\A}$-forest $\F$ to guess a set of edges $E'$ such that $\A\cup E'$ forms a TSP tour in $G[V_b\cup V(E')]$. Then, it constructs a special spanning forest $\F'$ and obtains a TSP tour by taking shortcuts on $\A\cup E'\cup2\F'$, similar to the FPT $3$-approximation algorithm for TSP parameterized by $p$, as mentioned earlier.

Next, we show the sub-algorithms {LIMB} and {CONNECT} in Steps~\ref{3.3} and~\ref{3.4}, the property of $\M$ in Step~\ref{3.5}, and the sub-algorithm {SHORTCUT} in Step~\ref{3.6}, respectively.

\subsubsection{The sub-algorithm: LIMB}\label{SC4.1.1}
Given $\A$ and $\F$, we assume that $T^*=\overline{a_1...b_1...a_{\size{\A}}...b_{\size{\A}}...}$, where $a_i...b_i\in\A$ for each $i$. 
Denote the number of anchors between $b_i$ and $a_{i+1}$ in $T^*$ by $f_i$, where $f_i\in\{1,2\}$. 
For each tree $F\in\F$, initialize a \emph{potential} (vertex) set $V_F\coloneqq\emptyset$.
LIMB works as follows.

First, we guess the values of $f_1,...,f_{\size{\A}}$, which specify the number of anchors between each pair $(b_i, a_{i+1})$ in $T^*$.
This allows us to determine the total number and relative positions of anchors in the ordering $\overline{a_1...b_1...a_{\size{\A}}...b_{\size{\A}}...}$.
For each such anchor, we can determine whether it is a single or pair anchor and how many bad vertices it connects to in $T^*$.
However, we do not know the exact identity of each anchor, i.e., which specific good vertex it corresponds to.

For example, consider the $T^*$ in Figure~\ref{fig1}. If we correctly guess $f_1=2$, $f_2=1$, and $f_3=2$, then we know that the anchors, in order, are pair, pair, single, pair, and pair. The corresponding sets of bad vertices incident to them are  $\{x_4\}$, $\{x_5\}$, $\{x_6,x_7\}$, $\{x_7\}$, $\{x_1\}$, respectively.

Next, for each anchor $x$, we guess the tree $F_x\in \F$ with $x\in V(F)$, select a set of potential vertices $V_x\subseteq V(F_x)$, and update the potential set by setting $V_{F_x}\coloneqq V_{F_x}\cup V_x$.
We will show that the size of each final potential set is $\O(q^2)$, and then, by guessing anchors in these potential vertex sets, we can obtain a set of limbs $\B'$ with desirable properties.

The details are put in Algorithm~\ref{3.1}.

\begin{algorithm}[t]
\caption{LIMB}
\label{alg3.1}
\textbf{Input:} An instance $G=(V=V_g\cup V_b, E, w)$, $\A$, and $\F$. \\
\textbf{Output:} A set of limbs $\B'$.
\begin{algorithmic}[1]
\State\label{3.1.3} Initialize $V_F\coloneqq\emptyset$ for each $F\in\F$, and $\B'=\emptyset$;

\State\label{3.1.1} Guess the values of $f_1,...,f_{\size{\A}}$.

\For{each anchor $x\in V_a$ in order $\overline{a_1...b_1...a_{\size{\A}}...b_{\size{\A}}...}$}\label{3.1.4}
\State\label{3.1.2} Guess the unique tree $F_x\in \F$ with $x\in V(F_x)$, and let $n_x=\min\{2q,\size{V(F_x)}\}$. \Comment{{We do not know the identity of $x$.}}
\If{$x$ is a single anchor}\label{3.1.5} \Comment{{This can be determined by the $f_i$ w.r.t.\ $x$.}}
\State\label{3.1.6} Let $y$ and $z$ be the bad vertices incident to $x$ in $T^*$.
\State\label{3.1.7} Obtain a set of $n_x$ vertices $V_x$ by selecting $x'\in V(F_{x})$ in the increasing order of $w(x',y)+w(x',z)$.
\Else\label{3.1.8}
\State\label{3.1.9} Let $y$ be the unique bad vertex incident to $x$ in $T^*$.
\State\label{3.1.10} Obtain a set of $n_x$ vertices $V_x$ by selecting $x'\in V(F_{x})$ in the increasing order of $w(x',y)$.
\EndIf\label{3.1.11}
\State\label{3.1.12} Update $V_{F_x}\coloneqq V_{F_x}\cup V_x$.
\EndFor\label{3.1.13}
\For{each anchor $x\in V_a$ in order $\overline{a_1...b_1...a_{\size{\A}}...b_{\size{\A}}...}$}\label{3.1.14}
\State\label{3.1.15} Guess $x$ by enumerating a vertex $x'\in V_{F_x}$.
\If{$x$ is a single anchor}\label{3.1.16}
\State\label{3.1.17} Let $y$ and $z$ be the bad vertices incident to $x$ in $T^*$.
\State\label{3.1.18} Update $\B'\coloneqq \B'\cup\{x'y,x'z\}$. \Comment{$x'y$ and $x'z$ are the limbs w.r.t.\ $x'$.}
\Else\label{3.1.19}
\State\label{3.1.20} Let $y$ be the unique bad vertices incident to $x$ in $T^*$.
\State\label{3.1.21} Update $\B'\coloneqq \B'\cup\{x'y\}$. \Comment{$x'y$ is the limb w.r.t.\ $x'$.}
\EndIf\label{3.1.22}
\EndFor\label{3.1.23}
\end{algorithmic}
\end{algorithm}

\begin{lemma}\label{thelimb}
There exists a configuration for the guessed anchors such that
\begin{enumerate}
    \item the guessed anchors are pair-wise distinct;
    \item for each anchor $x\in V_a$, the guessed anchor $x'$ for $x$ satisfies $x'=x$ if $x\in V_{F_x}$, and $x'\in V_x\subseteq V_{F_x}$ otherwise; 
    \item the set of all limbs $\B'$ w.r.t.\ the guessed anchors satisfies $w(\B')\leq w(\B)$.
\end{enumerate}
Moreover, {LIMB} can compute such a configuration in $\O(2^q\cdot q^{2q}\cdot (4q^2)^{2q})$ guesses, where each guess takes $n^{\O(1)}$ time.
\end{lemma}
\begin{proof}
First, assume that both Steps~\ref{3.1.1} and \ref{3.1.2} guess correctly.
Next, we create a configuration $\chi$ for the guessed anchors in Step~\ref{3.1.15} satisfying the three conditions in the lemma.

Let $V_a=V^0_a\cup V^1_a$, where $V^0_a=\{x\mid x\in V_a, x\notin V_x\}$, and $V^1_a=\{x\mid x\in V_a, x\in V_x\}$.
For each anchor $x$, let $x'$ denote the guessed anchor w.r.t.\ $x$. 
Since $x'$ must be guessed from the potential set $V_{F_x}$, we show how to set $x'$ in $\chi$.

We first consider the anchors in $V^1_a$.

\textbf{Case~1: $x\in V^1_a$.} In this case, we have $x\in V_x$ by definition. Since $V_x\subseteq V_{F_x}$ by Step~\ref{3.1.12}, we have $x\in V_{F_x}$. Thus, we simply set $x'=x$ in $\chi$, i.e., the anchor $x$ is guessed correctly in $\chi$.
Clearly, the weight of limbs w.r.t.\ $x'$ equals the weight of limbs w.r.t.\ $x$.

Then, we consider the anchors in $V^0_a$.

\textbf{Case~2: $x\in V^0_a$.} In this case, we have $x\notin V_x$ by definition. Since $V_x\subseteq V(F_x)$ by Step~\ref{3.1.12}, we have $V(F_x)\setminus V_x\neq\emptyset$, and then $\size{V(F_x)}>\size{V_x}$. By Steps~\ref{3.1.2}, \ref{3.1.7}, and \ref{3.1.10} we know that $\size{V_x}=2q$. Therefore, we have
\[
\size{V(F_x)}>2q.
\] 
Since $\size{V_a}\leq 2q$ by (\ref{eq3.1}), at least one vertex $x''\in V_x$ is still unused for guessing anchors, and then we set $x'=x''$ in $\chi$. 
By Steps~\ref{3.1.7} and \ref{3.1.10}, the weight of limbs w.r.t.\ $x'$ is at most the weight of limbs w.r.t.\ $x$.

Now, we have set all anchors in $V_a$ in $\chi$. Thus, $\chi$ is a valid configuration for the guessed anchors.

Under the above setting, the configuration $\chi$ clearly satisfies the first two conditions in this lemma. 
Since $\B'$ (resp., $\B$) is the set of all limbs w.r.t.\ all set anchors (resp., all real anchors in $V_a$), we have $w(\B')\leq w(\B)$ by definition.

Next, we analyze the total number of guesses (iterations) in {LIMB}. We begin with the following definition.

\begin{definition}
In an algorithm, we define the outermost loop as the \emph{level-1} loop, and let $n_i$ denote the number of iterations of the level-$i$ loop. In particular, each iteration of a level-$i$ loop performs $n_{i+1}$ iterations of the level-$(i+1)$ loop.
\end{definition}

Hence, LIMB performs $n_1\cdot n_2\cdots n_k$ iterations in total, where $k$ is the maximum level of loops.

We classify the loops in {LIMB} as follows.

\textbf{Level-1 Loop:} {LIMB} guesses the values of $f_1,...,f_{\size{\A}}$. Since $f_i\in\{1,2\}$ for each $i\leq \size{\A}$, and $\size{\A}\leq q$ by (\ref{eq3.1}), this incurs $\O(2^q)$ iterations.

\textbf{Level-2 Loop:} {LIMB} guesses the unique tree $F_x\in \F$ with $x\in V(F_x)$ for each anchor $x\in V_a$. Since $\size{V_a}\leq 2q$ by (\ref{eq3.1}) and $\size{\F}\leq q$ by (\ref{eq3.2}), this incurs $\O(q^{2q})$ iterations.

\textbf{Level-3 Loop:} {LIMB} guesses $x$ by enumerating a vertex from $V_{F_x}$ for each anchor $x\in V_a$. Since $\size{V_a}\leq 2q$ by (\ref{eq3.1}) and $\size{V_x}\leq 2q$ by Steps~\ref{3.1.7} and \ref{3.1.10} for each anchor $x\in V_a$, for each $F\in\F$, by Step~\ref{3.1.12}, we have $\size{V_F}\leq 2q\cdot 2q = 4q^2$. Thus, this incurs $\O((4q^2)^{2q})$ iterations.

Overall, {LIMB} performs $\O(2^q\cdot q^{2q}\cdot (4q^2)^{2q})$ iterations in total. In each iteration, {LIMB} clearly runs in $n^{\O(1)}$ time.
\end{proof}

\subsubsection{The sub-algorithm: CONNECT}\label{SC4.1.2}
The graph $\A\cup\B'\cup\F$ may be disconnected. For example, any tree $F\in\F$ with $V(F)\cap V_a=\emptyset$ must form a component (see (b) in Figure~\ref{fig3}). 

We introduce the sub-algorithm {CONNECT}, which augments $\A\cup\B'\cup\F$ into a connected graph $\A\cup\B'\cup\R'\cup\F$ by guessing a set of edges $\R'$.

It is easy to verify that adding the full set of good chains $\R$ to $\A \cup \B' \cup \F$ yields a connected graph. However, since $\size{V(\R)} = \size{V_g} = n - q$, enumerating all of $\R$ is not feasible in FPT time.
Therefore, in CONNECT, we instead guess the subset of trees $\F_{R}\subseteq\F$ connected by each $R\in\R$, as $\size{\F}\leq q$ by (\ref{eq3.2}). 
Note that this approach leads to $2^{\O(q^2)}$ guesses. To improve efficiency, we refine the strategy to reduce the number of guesses to $2^{\O(q \log q)}$.

First, we construct a complete graph $\widetilde{G}=(\widetilde{V},\widetilde{E},\widetilde{w})$ by contracting each tree $F\in\F$ into a vertex $u_F$, and define the weight between $u_F$ and $u_{F'}$ as 
\begin{equation}\label{eq3.3}
\widetilde{w}(u_F,u_{F'})=\min_{v\in V(F),v'\in V(F')}w(v,v').
\end{equation}

For each chain $R\in\R$, we obtain a minimal chain $R_m$ connecting all trees in $\F_R$ by shortcutting its internal vertices.
Then, we obtain a new chain $R^{*}$ by shortcutting the internal vertices, ensuring that each vertex in $\widetilde{V}_a\coloneqq\{u_F\mid F\in\F, V(F)\cap V_a\subseteq V^2_g, V(F)\setminus V_a\neq\emptyset\}$
appears in only one path in $\R^{*}\coloneqq\{R^{*}\mid R\in\R, \size{V(R)}\geq 2\}$.
It can be verified that the graph $\A\cup\B'\cup\R^{*}\cup\F$ remains connected.

Let $\F_{R^*}$ denote the set of trees in $\F$ connected by $R^*$, and define $\widetilde{V}_{R^*}\coloneqq\{u_F\mid F\in\F_{R^*}\}$.
Then, the collection $\widetilde{\V}\coloneqq\{\widetilde{V}_{R^{*}}\cap \widetilde{V}_a\mid R^*\in\R^*\}$ forms a partition of $\widetilde{V}_a$, which can be guessed in $2^{\O(q\log q)}$ times, since $\widetilde{V}_a$ is determined from $\B'$ and satisfies $\size{\widetilde{V}_a}\leq \size{\F}\leq q$ by (\ref{eq3.2}).

Once $\widetilde{\V}$ is guessed, we obtain $\widetilde{V}_{R^{*}}\cap \widetilde{V}_a$ for each $R\in \R$ with $\size{V(R)}\geq 2$.
Moreover, since $\widetilde{V}_{R^{*}}\setminus \widetilde{V}_a$ corresponds to trees in $\F$ containing the anchors of $R^*$, it is also determined by $\B'$. Hence, we obtain $\widetilde{V}_{R^*}$ as $\widetilde{V}_{R^*}=(\widetilde{V}_{R^{*}}\cap \widetilde{V}_a)\cup(\widetilde{V}_{R^{*}}\setminus \widetilde{V}_a)$.

We then compute an optimal TSP path or tour in $\widetilde{G}[\widetilde{V}_{R^*}]$ using the DP method~\cite{bellman1962dynamic}, which yields a set of edges in $G$, denoted by $E_R$, connecting all trees in $\F_{R^*}$ (see (e) in Figure~\ref{fig3}). 
By construction, $w(E_R) \leq w(R^*) \leq w(R)$.

Finally, we define $\R' = \bigcup_{R \in \R} E_R$, where $E_R = \emptyset$ if $\size{V(R)} = 1$.

The details are described in Algorithm~\ref{alg3.2}.

\begin{algorithm}[t]
\caption{CONNECT}
\label{alg3.2}
\textbf{Input:} An instance $G=(V=V_g\cup V_b, E, w)$, $\A$, $\B'$, and $\F$. \\
\textbf{Output:} A set of edges $\R'$.
\begin{algorithmic}[1]
\State Initialize $\R'=\emptyset$.

\State\label{alg52} Obtain a complete graph $\widetilde{G}=(\widetilde{V},\widetilde{E},\widetilde{w})$ by contracting each $F\in\F$ into a vertex $u_F$, where $\widetilde{w}$ is defined as in (\ref{eq3.3}).

\State\label{alg53} Set $\widetilde{V}_a=\{u_F\mid F\in\F, V(F)\cap V_a\subseteq V^2_g, V(F)\setminus V_a\neq\emptyset\}$.

\State\label{alg54} Guess the partition of $\widetilde{V}_a$, $\widetilde{\V}=\{\widetilde{V}_{R^{*}}\cap \widetilde{V}_a\mid R^*\in\R^*\}$, where $\widetilde{V}_{R^*}=\{u_F\mid F\in\F_{R^*}\}$, $\F_{R^*}$ is the set of trees in $\F$ connected by $R^*$, and $R^{*}$ is obtained by shortcutting the internal vertices of $R\in\R$, ensuring that each vertex in $\widetilde{V}_a$ appears in only one of the paths in $\R^{*}=\{R^{*}\mid R\in\R, \size{V(R)}\geq 2\}$.

\For{each good chain $R\in \R$ with $V(R)\geq 2$ (in order $T^*$)}\label{alg55}
\State Denote the pair anchors w.r.t.\ $R$ by $x$ and $x'$. \Comment{{We do not know $x$ and $x'$ but we know $F_x$ and $F_{x'}$ by $\B'$.}}

\State\label{connects1} Set $\widetilde{V}_{R^*}\coloneqq(\widetilde{V}_{R^*}\cap \widetilde{V}_a)\cup \{u_{F_x},u_{F_{x'}}\}$.

\State\label{connects2} Using the DP~\cite{bellman1962dynamic}, find a minimum-weight TSP tour in $\widetilde{G}[\widetilde{V}_{R^*}]$ if $F_x=F_{x'}$, and a minimum-weight TSP path between $u_{F_x}$ and $u_{F_{x'}}$ in $\widetilde{G}[\widetilde{V}_{R^*}]$ otherwise.

\State\label{alg59} Let $E_R$ be the corresponding set of edges in $G$ w.r.t.\ the obtained tour or path.
\State\label{alg510} Update $\R'\coloneqq \R'\cup E_R$.
\EndFor
\end{algorithmic}
\end{algorithm}

\begin{lemma}\label{lemmaCONNECT}
{CONNECT} can compute $\R'$, a set of edges between good vertices, in $2^{\O(q\log q)}$ guesses such that 
\begin{enumerate}
    \item $\A\cup\B'\cup\R'\cup\F$ forms a connected graph $G_\A$;
    \item $\size{Odd(G_\A)\cap V(F)}$ is even for each $F\in\F$;
    \item $Odd(G_\A)\cap V_b=\emptyset$;
    \item $w(\R')\leq w(\R)$.
\end{enumerate}
Moreover, each guess takes $\O(q^2\cdot 2^q)\cdot n^{\O(1)}$ time.
\end{lemma}
\begin{proof}
In {CONNECT}, there is only one level-1 loop, which guesses the partition $\widetilde{\V}=\{\widetilde{V}_{R^{*}}\cap \widetilde{V}_a\mid R^*\in\R^*\}$ of $\widetilde{V}_a$ in Step~\ref{alg54}.

We first consider the number of iterations in this loop.

In the partition, we should also know the corresponding set $\widetilde{V}_{R^*}\cap \widetilde{V}_a$ for each $R\in\R$. 
We can obtain all such possible partitions as follows.

First, initialize a set $\widetilde{V}'_R$ for each $R\in\R$. Then, enumerate all cases of partitioning every element in $\widetilde{V}_a$ into exactly one of the sets in $\{\widetilde{V}'_R\mid R\in\R\}$. 
This incurs at most $\size{\R}^{\size{\widetilde{V}_a}}$ distinct iterations.
By (\ref{3.1}), we have $\size{\R}\leq q$. By (\ref{eq3.2}) and Steps~\ref{alg52} and \ref{alg53}, we know 
\begin{equation}\label{sizeva}
\size{\widetilde{V}_a}\leq\size{\F}\leq q.
\end{equation} 
Therefore, we have at most $\size{\R}^{\size{\widetilde{V}_a}}\leq q^{q}\subseteq 2^{O(q\log q)}$ distinct iterations.

In each iteration, the running time of {CONNECT} is dominated by computing a TSP tour or path in $\widetilde{G}[\widetilde{V}_{R^*}]$ by using the DP method~\cite{bellman1962dynamic,held1962dynamic} in Step~\ref{connects2} for each good chain $R\in\R$ with $\size{V(R)}\geq 2$. Thus, the running time is at most
\[
\sum_{R^*\in\R^*}\O(\size{\widetilde{V}_{R^*}}^2\cdot 2^{\size{\widetilde{V}_{R^*}}})\cdot n^{\O(1)}.
\]
Since $\widetilde{\V}=\{\widetilde{V}_{R^{*}}\cap \widetilde{V}_a\mid R^*\in\R^*\}$ forms a partition of $\widetilde{V}_a$, we have 
$\sum_{R^*\in\R^*}\size{\widetilde{V}_{R^*}\cap \widetilde{V}_a}=\size{\widetilde{V}_a}\leq q$ by (\ref{sizeva}).
In addition,  
$\size{\widetilde{V}_{R^*}}\leq \size{\widetilde{V}_{R^*}\cap \widetilde{V}_a}+2$ by Step~\ref{connects1}, and then $\sum_{R^*\in\R^*}\size{\widetilde{V}_{R^*}}=\O(q)$.
Thus, we have
\[
\sum_{R^*\in\R^*}\O(\size{\widetilde{V}_{R^*}}^2\cdot 2^{\size{\widetilde{V}_{R^*}}})=\O(q^2\cdot 2^q).
\] 
Therefore, each iteration takes $\O(q^2\cdot 2^q)\cdot n^{\O(1)}$ time.

Then, we prove the four properties of $\R'$ in this lemma. 

First, we prove that $G_\A=\A\cup\B'\cup\R'\cup\F$ is connected.

Since $T^*=\A\cup\B\cup\R$, we know that $\A\cup\B\cup\R$ forms a connected graph. By Lemma~\ref{thelimb}, for any limb $vx\in \B$, where $v$ is a bad vertex and $x$ is an anchor, the guessed anchor $x'$ for $x$ satisfies $x'\in V_x\subseteq V_{F_x}$, and there exists a limb $vx'\in \B'$ connecting $v$ to $F_x$. 
Thus, $\A\cup\B'\cup\R\cup\F$ remains connected. 

For each good chain $R\in\R$ with $V(R)\geq 2$, it connects all trees in $\F_R$. 
Recall that $\R^*$ is obtained by shortcutting the internal vertices of chains in $\R$, ensuring that each vertex in $\widetilde{V}_a$ appears in only one of the paths in $\R^{*}$.
Thus, $\A\cup\B'\cup\R^*\cup\F$ remains a connected graph.

By Step~\ref{alg54}, we have 
\[
\R^{*}=\{R^{*}\mid R\in\R, \size{V(R)}\geq 2\}.
\]
Moreover, by Steps~\ref{alg55}-\ref{alg510}, we know 
\[
\R'=\{E_R\mid R^*\in \R^*\}.
\]
Specifically, for each $R^*\in\R^*$, the edges in $E_R$ connects all trees in $\F_{R^*}$.
Thus, $G_\A=\A\cup\B'\cup\R'\cup\F$ is connected.

Hence, the first property holds.

Then, we prove that $\size{Odd(G_\A)\cap V(F)}$ is even for each $F\in\F$.

Consider the graph $G_\A/\F$, which is obtained by contracting all vertices in $V(F)$ into a single vertex $F$ for each $F\in\F$. The edges of $G_\A/\F$ are defined as follows: (1) include all edges in $\A$; (2) for any edge $vx\in\B'$, where $v\in V_b$ and $x\in V(F)$, add an edge $vF$ to $G_\A/\F$; (3) for any edge $xx'\in\R'$, where $x\in V(F)$ and $x'\in V(F')$ and $F\neq F'$, add an edge $FF'$ to $G_\A/\F$. 

By LIMB, each bad vertex is a 2-degree vertex in $A\cup\B'$, and thus it remains a 2-degree vertex in $G_\A/\F$.  
By Step~\ref{connects2} of CONNECT, each contacted vertex $F\in\F$ is a vertex of even degree in $G_\A/\F$. Moreover, since $G_\A=\A\cup\B'\cup\R'\cup\F$ is connected, the graph $G_\A/\F$ is also connected. Therefore, $G_\A/\F$ is an Eulerian graph (see (e) in Figure~\ref{fig3}).

Now, we consider the graph $G_\A$. Let $d(v)$ denote the degree of vertex $v$ in $G_\A$.

Since $G_\A/\F$ is Eulerian, we know that $Odd(G_\A)\subseteq V_g$. 
For any $F\in\F$, let $m_{out}$ denote the number of edges with one vertex in $V(F)$ and the other vertex in $V\setminus V_F$ in $G_\A$, and $m_{in}$ the number of edges with both endpoints in $V(F)$ in $G_\A$. It is clear that
\begin{equation}\label{inout}
\sum_{v\in V(F)}d(v) = m_{out} + 2\cdot m_{in}.
\end{equation}
Since $G_\A/\F$ is an Eulerian graph, $m_{out}$ must be even. 
Thus, by (\ref{inout}), we know that $\sum_{v\in V(F)}d(v)$ is even, which implies that the number of odd-degree vertices in $V(F)$ is even, i.e., $\size{Odd(G_\A)\cap V(F)}$ is even.

Hence, the second property holds.

Next, we prove that $Odd(G_\A)\cap V_b=\emptyset$.

As mentioned above, $Odd(G_\A)\subseteq V_g$, i.e., $Odd(G_\A)\cap V_b=\emptyset$.
In fact, since $\R'$ and $\F$ only contain edges between good vertices, each bad vertex in $G_\A$ is a 2-degree vertex. 

Hence, the third property holds.

Finally, we prove that $w(\R')\leq w(\R)$.

By the definition of $\F_R$, for each good chain $R\in\R$, it connects all trees in $\F_R$. Recall that the chain $R^*$ is obtained by shortcutting the internal vertices in $R$, and thus it satisfies $w(R^*)\leq w(R)$. 

Moreover, since $R^*$ corresponds to a TSP tour or path in $\widetilde{G}[\widetilde{V}_{R^*}]$ and {CONNECT} uses the DP method to compute a minimum-weight TSP tour or path, by Step~\ref{connects2}, we know that $E_R$, the set of edges w.r.t.\ the TSP tour or path, has weight at most 
\[
w(E_R)\leq w(R^*)\leq w(R).
\]
Therefore, by the definition of $\R'$ we have 
\begin{align*}
w(\R')&=\sum_{R\in \R, V(R)\geq2}w(E_R)\leq \sum_{R\in \R, V(R)\geq2}w(R)\leq \sum_{R\in \R}w(R) = w(R).   
\end{align*}

Hence, the fourth property holds.

Therefore, the lemma holds.
\end{proof}

\subsubsection{The sub-algorithm: SHORTCUT}\label{SC4.1.3}
Given $G_\A=\A\cup\B'\cup\R'\cup\F$, Step~\ref{3.5} of ALG.4 constructs an Eulerian graph $G'_\A=\A\cup\B'\cup\R'\cup\F\cup\M$ by using $\M$ to fix the degree parity of the vertices in $Odd(G_\A)$.
By Lemma~\ref{lemmaCONNECT}, this step is feasible. 

To analyze the quality of $\M$, we use the following lemma.

\begin{lemma}\label{usecoro}
It holds that $w(\M)\leq w(\F)$.
\end{lemma}
\begin{proof}
We begin by proving the following well-known claim (see, e.g., Lemma~19 in \cite{DBLP:journals/mp/BlauthTV23}).

\begin{claim}\label{MwithF}
Let $T$ be any tree and let $S \subseteq V(T)$ be a subset of even cardinality. Then any minimum-weight matching $\M_S$ on $S$ satisfies $w(\M_S) \leq w(T)$.
\end{claim}
\begin{proof}[Claim Proof]
Since $\size{S}\bmod 2=0$, by arbitrarily pairing all vertices in $S$, we obtain a set of paths $\B$, where each $v\sim v'\in\B$ denotes the unique path between $v$ and $v'$ on $T$. 
Assume that there exist two paths $v...xy...v',u...xy...u'\in\B$ sharing at least one edge $xy$. We can modify these two paths into the paths $v\sim u$ and $v'\sim u'$, in such a way that the total number of edges on these two paths is reduced by at least 1. 
Therefore, by repeatedly performing the above procedure, we can ensure that $\B$ is a set of edge-disjoint paths. 

By the triangle inequality, we have $w(v,v')\leq w(v\sim v')$ for each $v\sim v'\in \B$. Thus, there is a matching $\{vv'\mid v\sim v'\in\B\}$ on $S$ with weight at most $w(\B)\leq w(T)$. Since $\M_S$ is a minimum-weight matching on $S$, we have $w(M_S)\leq w(T)$.
\end{proof}

Then, we are ready to prove Lemma~\ref{usecoro}.

By Step~\ref{3.5} of ALG.4, $\M$ is obtained by finding a minimum-weight matching $M_F$ in the graph $G[Odd(G_\A)\cap V(F)]$ for each $F\in\F$. This is feasible since $Odd(G_\A)\cap V(F)$ is even for any $F\in\F$ by Lemma~\ref{lemmaCONNECT}.
By Claim~\ref{MwithF}, we obtain $w(M_F)\leq w(F)$ for each $F\in\F$.
Thus, we have $w(\M)\leq w(\F)$.
\end{proof}

Since $G'_\A$ is Eulerian, we use the sub-algorithm {SHORTCUT} to take shortcuts on $G'_\A$ to obtain a TSP tour $T_3$ in $G$ with a non-increasing weight.

In $G'_\A$, each guessed single anchor connects exactly two bad chains. Since a triangle containing one bad vertex may violate the triangle inequality, in {SHORTCUT}, we first construct a new Eulerian graph $G''_\A$ by deleting some edges in $\M\cup\F$ and duplicating some guessed single anchors. Then, we can use a similar idea in Lemma~\ref{LB2.3} to take shortcuts on $G''_\A$ to obtain $T_3$ with a non-increasing weight.

The graph $G''_\A$ is constructed as follows. 
For each $F\in\F$ where $V(F)$ includes a guessed single anchor, we consider two cases.

\textbf{Case 1:} If $V(F)$ consists only of guessed single anchors, then it contains neither internal vertices nor pair anchors, implying that no edge in $\R'$ connects to $F$. In this case, we simply delete all edges in $E(F)\cup M_F$, which does not affect connectivity. Each guessed single anchor is then labeled as a \emph{marked bad vertex} (which is still a good vertex).


\textbf{Case 2:} If $V(F)$ contains at least one internal vertex or pair anchor, then there exists an edge in $\R'$ or a limb in $\B'$ connecting $F$. 
For each guessed single anchor $x\in V(F)$, let its associated limbs be $xy$ and $xz$, where $y, z$ are bad vertices. 
We create a copy $x'$ of $x$, remove $xy$ and $xz$, and add $x'y$ and $x'z$. We refer to $x'$  as a marked bad vertex.

These two cases are illustrated in Figure~\ref{fig4}, and the detailed procedure is given in Algorithm~\ref{alg3.3}.

\begin{figure}[t]
    \centering
    
    \begin{subfigure}[b]{0.8\linewidth}
        \centering
        \includegraphics[width=0.5\linewidth]{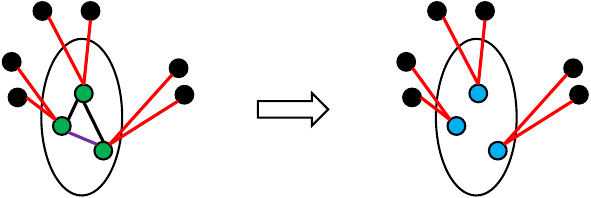}
        \caption{$V(F)$ consists only of guessed single anchors.}
    \end{subfigure}
    
    \vspace{1mm}
    
    \begin{subfigure}[b]{0.8\linewidth}
        \centering
        \includegraphics[width=0.5\linewidth]{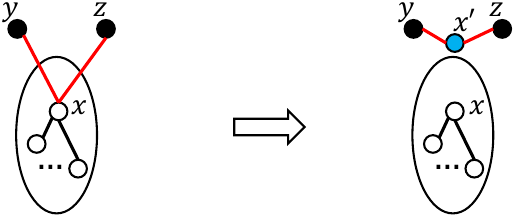}
        \caption{$V(F)$ contains one guessed single anchor and at least one internal vertex or pair anchor.}
    \end{subfigure}
    
    \caption{An illustration of the two cases. Each black node denotes a bad vertex, each green or white node denotes a good vertex, each blue node denotes a marked bad vertex; the edges in $E(F)$ are shown in black, the edges in $\M_F$ are shown in purple, and the limbs are shown in red. In (a), the black and purple edges are deleted; In (b), $x'$ is a copy of $x$.}
    \label{fig4}
\end{figure}

\begin{algorithm}[t]
\caption{SHORTCUT}
\label{alg3.3}
\textbf{Input:} $G=(V=V_g\cup V_b, E, w)$, $\A$, $\B'$, $\R'$, $\F$, and $\M$. \\
\textbf{Output:} A TSP tour $T_3$ in $G$.
\begin{algorithmic}[1]

\For{each $F\in\F$ such that $V(F)$ contains one guessed single anchor}
\If{$V(F)$ consists only of guessed single anchors}
\State\label{mody1} Delete all edges in $E(F)\cup M_F$, and refer to all guessed single anchors as marked bad vertices.
\Else\label{mody2}
\For{each guessed single anchor $x\in V(F)$}\label{singleanchor}\label{mody21}
\State\label{mody22} Denote the limbs w.r.t.\ $x$ by $xy$ and $xz$.
\State\label{mody23} Create a marked bad vertex $x'$.
\State\label{mody24} Replace edges $xy$ and $xz$ with $x'y$ and $x'z$.
\EndFor
\EndIf
\EndFor
\State Denote the above graph by $G''_\A$.
\State Obtain a TSP tour $T_3$ in $G$ (see Lemma~\ref{alg3shortcut}).
\end{algorithmic}
\end{algorithm}


\begin{lemma}\label{alg3shortcut}
Given $G''_\A$, there is an $\O(n^2)$-time algorithm to obtain a TSP tour $T_3$ in $G$ with weight at most $w(G''_\A)$. 
\end{lemma}
\begin{proof}
Recall that SHORTCUT may introduce marked bad vertices. Although a marked bad vertex may in fact be a good vertex, we treat all marked bad vertices as bad vertices in the following.

We begin by proving the following two claims.

\begin{claim}\label{alg3eulerian}
The graph $G''_\A$ remains Eulerian, where each bad vertex is a 2-degree vertex and each good vertex is incident to at most one bad vertex.
\end{claim}
\begin{proof}[Claim Proof]
We first prove that in the graph $G''_\A$, each bad vertex is a 2-degree vertex and each good vertex is incident to at most one bad vertex.

In the graph $G_\A=\A\cup\B'\cup\R'\cup\F$, as shown in the proof of Lemma~\ref{lemmaCONNECT}, each bad vertex is a 2-degree vertex. Moreover, by LIMB, each good vertex is incident to at most two bad vertices, and each good vertex that is incident to two bad vertices must be a guessed single anchor.

Since $\M$ consists only of edges between good vertices, the graph $G'_\A=\A\cup\B'\cup\R'\cup\F\cup\M$, obtained by adding the edges in $\M$ to $G_\A$, also satisfies the above two properties.

By SHORTCUT, the graph $G''_\A$ is obtained by modifying all guessed single anchors in $G'_\A$.
Specifically, if $V(F)$ consists only of guessed single anchors, by Step~\ref{mody1}, all guessed single anchors become 2-degree bad vertices.
Otherwise, by Steps~\ref{mody21}-\ref{mody24}, each guessed single anchor in $V(F)$ is modified into a good vertex (connecting no bad vertices), and a 2-degree bad vertex. 

Thus, in the graph $G''_\A$, each bad vertex is a 2-degree vertex, and each good vertex is incident to at most one bad vertex. In fact, each good vertex that is incident to one bad vertex must be a pair anchor.

Next, we prove that $G''_\A$ is Eulerian.

By Lemma~\ref{lemmaCONNECT} and Step~\ref{3.5} of ALG.4, $G'_\A$ is Eulerian, and thus each vertex in $G'_\A$ has even degree.
Since the modifications in SHORTCUT do not change the parity of any vertex's degree, each vertex in $G'_\A$ still has even degree.

Then, we only need to prove that $G''_\A$ is connected.

Assume that $G''_\A$ is not connected. 
By SHORTCUT, $G''_\A$ still contains all bad chains in $\A$.
Thus, there exist at least two connected components in $G''_\A$, denoted as $C_1$ and $C_2$, such that $C_1$ contains at least one bad chain in $\A$ and $C_2$ contains no bad chains, only if all bad chains are contained in $C_1$. Moreover, we can choose $C_2$ such that when $C_2$ contains one bad chain, there exists a good chain $R\in\R$ that connects one bad chain in $C_1$ and one bad chain in $C_2$.

Next, we consider the following two cases.

\textbf{Case~1:} $C_2$ contains no bad chains in $\A$. In this case, all bad chains are contained in $C_1$. By SHORTCUT, a marked bad vertex connects two bad chains in $\A$, and a pair anchor still connects one bad chain in $\A$. Then, $C_2$ must contain a tree $F\in\F$ such that $V(F)$ contains no pair anchors. Moreover, by Step~\ref{mody2}, $V(F)$ must contain an internal vertex. 

Let $\F^*$ denote the set of trees in $\F$ that do not consist only of single guessed anchors before SHORTCUT.
Then, by CONNECT, $\R'\cup\F^*$ ensures that all internal vertices must be connected to some pair anchor, say $x$. This property still holds even after SHORTCUT by Steps~\ref{mody21}-\ref{mody24}.
By LIMB, there must be a limb connecting $x$ to a bad vertex.
However, since $C_1$ contains all bad chains in $\A$, the internal vertex in $C_2$ must connect some bad vertex in $C_1$, which contradicts the condition that $C_1$ and $C_2$ are not connected.


\textbf{Case~2:} $C_2$ contains at least one bad chain in $\A$.
In this case, there exists one good chain $R\in\R$ that connects one bad chain in $C_1$ and one bad chain in $C_2$. Assume that $R=x...x'\in\R$. Then, similar to the analysis in Case~1, by CONNECT, $\R'\cup\F^*$ ensures that $x$  must be connected to $x'$. This property still holds even after SHORTCUT by Steps~\ref{mody21}-\ref{mody24}. Thus, this contradicts the condition that $C_1$ and $C_2$ are not connected.

Therefore, $G''_\A$ is an Eulerian graph.
\end{proof}

\begin{claim}\label{usefulclaim}
Given any Eulerian graph $G'$ with $\O(m)$ edges, where each bad vertex has degree 2 and each good vertex is incident to at most one bad vertex, there exists an $\O(m)$-time algorithm to compute a TSP tour $T'$ with  $w(T')\leq w(G')$. 
Moreover, any good vertex that is not incident to a bad vertex in $G'$ remains not incident to any bad vertex in $T'$.
\end{claim}
\begin{proof}[Claim Proof]
Since $G'$ is Eulerian and consists of $\O(m)$ edges, an Eulerian tour $T$ can be computed in $\O(m)$ time~\cite{cormen2022introduction}.

Each bad vertex has degree 2, so it appears only once in $T$. 
Thus, if a vertex $a$ appears multiple times in $T$, it must be a good vertex.
Moreover, since each good vertex is incident to at most one bad vertex in $G'$, every appearance of $a$ in $T'$, except one, is incident to two good vertices. Each such appearance can be shortcut without increasing the weight. 

Therefore, we can obtain a TSP tour $T'$ in $G'$ with weight at most $w(G')$, and the total running time is $\O(m)$.

When we take a shortcut in this way, the removed edges and the newly added edge are all between good vertices. Thus, the set of edges between a good vertex and a bad vertex in $T$ remains unchanged in $T'$. 
Therefore, any good vertex that is not incident to a bad vertex in $G'$ remains not incident to any bad vertex in $T'$.
\end{proof}

Then, we are ready to prove Lemma~\ref{alg3shortcut}.

By LIMB, $\A$, $\B'$, and $\F$ each contains at most $n$ edges. 
By the proof of Lemma~\ref{lemmaCONNECT}, $\R'$ contains at most $\O(q)$ edges. By Step~\ref{3.5} of ALG.4, $\M$ contains at most $n$ edges. 
Note that $q\leq n$.
Therefore, the graph $G'_\A=\A\cup\B'\cup\R'\cup\F\cup\M$ contains at most $\O(n)$ edges. Then, by SHORTCUT, the graph $G''_\A$ also contains at most $\O(n)$ edges.

By Claims~\ref{alg3eulerian} and \ref{usefulclaim}, we can take shortcuts on $G''_\A$ to obtain a TSP tour $T$ in $G''_\A$ such that $w(T)\leq w(G''_\A)$ in $\O(n)$ time. 
Note that $T$ is not a TSP tour in $G$, since for each guessed single anchor $x$ in Step~\ref{singleanchor} of SHORTCUT, $T$ contains both $x$ and its copy $x'$. 
By {SHORTCUT}, $x$ is incident only to good vertices in $G''_\A$, and thus remains incident only to good vertices in $T$ by Claim~\ref{usefulclaim}. Thus, we can keep $x'$ and shortcut $x$ without increasing the weight.
Therefore, given $T$, we can further obtain a TSP tour $T_3$ in $G$ such that 
\[
w(T_3)\leq w(T)\leq w(G''_\A).
\]
The overall running time is $\O(n)\subseteq\O(n^2)$.
\end{proof}

\begin{theorem}\label{t3}
For TSP parameterized by $q$, ALG.4 achieves an FPT $3$-approximation ratio with running time $2^{\O(q\log q)}\cdot n^{\O(1)}$.
\end{theorem}
\begin{proof}
First, we analyze the weight of the TSP tour $T_3$ in $G$.

By {SHORTCUT} and Lemma~\ref{alg3shortcut}, we know that 
\begin{align*}
w(T_3)&\leq w(G''_\A)\leq w(G'_\A)=w(\A)+w(\B')+w(\R')+w(\F)+w(\M).
\end{align*}
Then, by (\ref{eq3.0}), (\ref{eq3.2}), and Lemmas~\ref{thelimb}, \ref{lemmaCONNECT}, and \ref{usecoro}, we obtain 
\begin{align*}
w(\A)+w(\B')+w(\R')+w(\F)+w(\M)
&\leq w(\A)+w(\B)+w(\R)+w(\F)+w(\F)\\
&\leq w(\A)+w(\B)+3w(\R)\\
&\leq 3(w(\A)+w(\B)+w(\R))=3\cdot\OPT.
\end{align*}
Thus, we have $w(T_3)\leq 3\cdot\OPT$, which implies that ALG.4 achieves a $3$-approximation ratio.

Next, we analyze the overall running time of ALG.4.

We classify the loops in ALG.4 as follows.

\textbf{Level-1 Loop:} Step~\ref{3.1} guesses the set of bad chains $\A$, incurring $\O(2^q\cdot q!)$ iterations, each taking $n^{\O(1)}$ time~\cite{arxiv24}.

\textbf{Level-2 Loop:} Step~\ref{3.3} guesses the set of limbs $\B'$ satisfying the properties in Lemma~\ref{thelimb}, incurring $\O(2^q\cdot q^{2q}\cdot (4q^2)^{2q})$ iterations, each taking $n^{\O(1)}$ time. 

\textbf{Level-3 Loop:} Step~\ref{3.4} guesses the set of edges  $\R'$ satisfying the properties in Lemma~\ref{lemmaCONNECT}, incurring $2^{\O(q\log q)}$ iterations, each taking $\O(q^2\cdot 2^q)\cdot n^{\O(1)}$ time. 

Moreover, Step~\ref{3.2} takes $\O(n^2\log n)$ time~\cite{khachay2016approximability}, Step~\ref{3.5} takes $\O(n^3)$ time~\cite{lawler1976combinatorial}, and Step~\ref{3.6} takes $\O(n^2)$ time by Lemma~\ref{alg3shortcut}.

Putting all together, the total number of iterations is
\[
\O(2^q\cdot q!)\cdot\O(2^q\cdot q^{2q}\cdot (4q^2)^{2q})\cdot 2^{\O(q\log q)}\subseteq 2^{\O(q\log q)},
\]
and the running time per iteration is
\begin{align*}
&n^{\O(1)}+\O(n^2\log n)+n^{\O(1)}+\O(q^2\cdot 2^q)\cdot n^{\O(1)}+\O(n^3)+\O(n^2)\subseteq 2^{\O(q)}\cdot n^{\O(1)}.
\end{align*}

Therefore, ALG.4 achieves an FPT $3$-approximation ratio with overall running time 
\[
2^{\O(q\log q)}\cdot 2^{\O(q)}\cdot n^{\O(1)}\subseteq 2^{\O(q\log q)}\cdot n^{\O(1)},
\]
as desired.
\end{proof}

\section{Conclusion}
Many problems admit good approximation algorithms on metric graphs but become hard to approximate on general graphs. In practice, due to measurement errors in data or the intrinsic nature of certain problems, graphs that are originally metric may locally violate the triangle inequality. Although such violations are typically rare, they may render algorithms designed for metric graphs inapplicable. Recently, there has been growing interest in characterizing ``distance'' parameters between general and metric graphs, and in leveraging these parameters to develop parameterized algorithms.

We investigate FPT approximation algorithms for the classic TSP problem under two parameters, $p$ and $q$. By introducing new techniques, we significantly improve upon the best-known FPT approximation algorithms for both parameters. Our methods also show promise for extension to related problems parameterized by $p$ and $q$, such as the asymmetric TSP~\cite{DBLP:journals/siamcomp/TraubV22} and the prize-collecting TSP~\cite{DBLP:conf/stoc/BlauthN23}. 

Several directions are worthy of further investigation:
\begin{enumerate}
\item 
For the parameter $q$, is it possible to design a constant-factor approximation algorithm with running time $2^{\mathcal{O}(q)} \cdot n^{\mathcal{O}(1)}$?

\item  It would also be valuable to explore additional parameters that quantify the ``distance'' from a general graph to a metric one, potentially enabling better algorithms for handling near-metric graphs.
\end{enumerate}


\bibliographystyle{plain}
\bibliography{aaai2026}

\appendix

\section{Two Assumptions}\label{appa}
\subsection{TSP Parameterized by $p$}
Recall that we assume $\min\{\size{V_g}, \size{V_b}\}\geq 3$.
We now show that when $\min\{\size{V_g}, \size{V_b}\}<3$, one can easily obtain an FPT $\alpha$-approximation algorithm with running time $2^{\O(p)}+n^{\O(1)}$. 

\textbf{Case~1: $\size{V_g}<3$.} In this case, we can directly solve the problem optimally in $2^{\O(p)}$ time using the DP method~\cite{bellman1962dynamic}. Thus, the problem admits an FPT approximation ratio of $1$ with running time $2^{\O(p)}$.

\textbf{Case~2: $\size{V_b}<3$.} In this case, the number of violating triangles must be 0, and then $\size{V_b}=0$, i.e., the problem reduces to metric TSP~\cite{DBLP:conf/stoc/KarlinKG21}. Thus, the problem admits an $\alpha$-approximation algorithm with running time $n^{\O(1)}$. 

Hence, when $\min\{\size{V_g}, \size{V_b}\}<3$, the problem admits an FPT $\alpha$-approximation algorithm with running time $2^{\O(p)}+n^{\O(1)}$.

\subsection{TSP Parameterized by $q$}
Recall that we assume $\min\{\size{V_g}, \size{V_b}\}\geq 1$. Using a similar argument as in the case of $p$, it is easy to obtain that when $\min\{\size{V_g}, \size{V_b}\}<1$, the problem admits an FPT $\alpha$-approximation algorithm with running time $2^{\O(q)}+n^{\O(1)}$.

\end{document}